\documentclass{article}
\usepackage{natbib}
\usepackage{authblk}
\usepackage{amsmath,amsthm,amsfonts,amssymb,color,hyperref,anysize,enumitem,graphicx,epstopdf,algorithm,algpseudocode,cleveref,multirow}
\usepackage[usenames,dvipsnames]{xcolor}
\usepackage[margin=0.6in]{geometry}

\newcommand{\rr}{\mathbb{R}}

\newcommand{\nn}{\mathbb{N}}

\newcommand{\prob}[1]{\mathbb{P}\left[ #1 \right]}

\newcommand{\expect}[1]{\mathbb{E}\left[ #1 \right]}

\DeclareMathOperator*{\argmin}{arg\,min}
\DeclareMathOperator*{\argmax}{arg\,max}

\usepackage{framed}

\newcommand{\boxlem}[1]{
\begin{framed}
\setlength{\topsep}{1pt}
\begin{lemma}
\normalfont #1
\end{lemma}
\end{framed}
}

\usepackage{tcolorbox}
\tcbuselibrary{theorems}

\newtcbtheorem[auto counter]{coloredDEF}{Definition}
{colback=blue!5,colframe=blue!40!black,fonttitle=\bfseries,before={\vspace{0.3cm}},after={\vspace{0.3cm}}}{de}

\newtcbtheorem[auto counter]{coloredNOT}{Notation}
{colback=blue!5,colframe=blue!40!black,fonttitle=\bfseries,before={\vspace{0.3cm}},after={\vspace{0.3cm}}}{no}

\newtcbtheorem[auto counter]{assume}{Assumption}
{colback=green!5,colframe=green!60!black!80,fonttitle=\bfseries,before={\vspace{0.3cm}},after={\vspace{0.3cm}}}{assume}

\newtcbtheorem[auto counter]{coloredLEM}{Lemma}
{colback=red!5,colframe=red!60!black,fonttitle=\bfseries,before={\vspace{0.3cm}},after={\vspace{0.3cm}}}{lm}

\newtcbtheorem[auto counter]{coloredTHM}{Theorem}
{colback=purple!5,colframe=purple!50!black,fonttitle=\bfseries,before={\vspace{0.3cm}},after={\vspace{0.3cm}}}{th}

\newtcbtheorem[auto counter]{coloredPROP}{Proposition}
{colback=purple!5,colframe=purple!50!black,fonttitle=\bfseries,before={\vspace{0.3cm}},after={\vspace{0.3cm}}}{pr}

\newtcbtheorem[auto counter]{coloredCOR}{Corollary}
{colback=purple!5,colframe=purple!50!black,fonttitle=\bfseries,before={\vspace{0.3cm}},after={\vspace{0.3cm}}}{co}

\newtheorem{theorem}{Theorem}
\newtheorem{lemma}[theorem]{Lemma}
\newtheorem{corollary}[theorem]{Corollary}

\newtheorem{example}[theorem]{Example}

\newtheorem{prop}[theorem]{Proposition}
\newtheorem{definition}[theorem]{Definition}
\newtheorem{shortexercise}[theorem]{Short Exercise}
\newtheorem{exercise}[theorem]{Exercise}

\newcommand{\calC}{\mathcal{C}}

\newcommand{\calI}{\mathcal{I}}

\newcommand{\calP}{\mathcal{P}}
\newcommand{\calQ}{\mathcal{Q}}

\newcommand{\calV}{\mathcal{V}}

\newcommand{\bbc}{\mathbf{c}}

\newcommand{\bbq}{\mathbf{q}}

\newcommand{\bbs}{\mathbf{s}}

\newcommand{\bbv}{\mathbf{v}}

\newcommand{\bbx}{\mathbf{x}}
\newcommand{\bby}{\mathbf{y}}
\newcommand{\bbz}{\mathbf{z}}

\newcommand{\bphi}{\boldsymbol{\phi}}

\newcommand{\creators}{C}
\newcommand{\hide}[1]{}

\newcommand{\inner}[2]{\left\langle ~#1 ~,~ #2~\right\rangle}

\usepackage{graphicx} %
\title{Dynamics in Two-Sided Attention Markets: \\Objective, Optimization, and Control}
\author{Haiqing Zhu\footnote{Email addresses: \{haiqing.zhu, yunkuen.cheung, lexing.xie\}@anu.edu.au }}
\author{Yun Kuen Cheung}
\author{Lexing Xie}
\affil{School of Computing, Australian National University}
\date{August 2025}

\begin{document}
\maketitle

\begin{abstract}
With most content distributed online and mediated by platforms, there is a pressing need to understand %
the ecosystem of content creation and consumption.
A considerable body of recent work shed light on the one-sided market on creator-platform or user-platform interactions, showing key properties of static (Nash) equilibria and online learning.
In this work, we examine the {\it two-sided} market including the platform and both users and creators. 
We design a potential function for the coupled interactions among users, platform and creators.
We show that such coupling of creators' best-response dynamics with users' multilogit choices is equivalent to mirror descent on this potential function. 
Furthermore, a range of platform ranking strategies correspond to a family of potential functions, and the dynamics of two-sided interactions still correspond to mirror descent. We also provide new local convergence result for mirror descent in non-convex functions, which could be of independent interest.
Our results provide a theoretical foundation for explaining the diverse outcomes observed in attention markets.

\end{abstract}

\section{Introduction}
We are concerned with {\it attention markets} of online content -- where 
the key resource being allocated is users' attention (and time). 
This market is {\it two-sided}. On one-side, a platform distributes content to users, controlling content {\it visibility} via recommender systems; on the other side, the platform syndicates content from creators and distributes reward back to them. Feedback loops tie the two sides together. 
Our main concern and contribution lies in an optimization view of such two-sided interactions: 
is there a global objective, do myopic updates minimise any objective, and 
do different recommendation strategies change the overall optimisation problem? 

The answers to all three question are affirmative, but we shall first %
relate our questions to what has been asked about
markets.
In {\it attention markets}, we focus on the allocation of users' time and creators' effort, neither of which is quantified in monetary terms but is the key driver to most monetary incentives. This is distinct from the large body of work on advertising-driven attention~\cite{varian2009online}, or charging schemes on either side~\cite{Vogelsang2010DynamicsOT}.
In economics, markets in which platforms act as intermediaries between two groups of agents are known as \emph{two-sided markets}\footnote{These markets are different from the ``two-sided markets'' in bipartite matching context.}.
Our concern in {\it two-sided markets} is a special case with one platform mediating a large number of players on either side, and hence distinct from known literature on platform competition~\cite{RochetTirole-twosided,Armstrong-twosided}.
Furthermore, price is not the key mediator in our market, but {\it bounded attention} is. We assume that the total attention available is bounded rather than infinite~\cite{GhoshHummel2014}, which is arguably closer to the original notion of attention markets~\cite{simon1971designing}. In comparison, there is {\it unbounded supply} of each digital good, assuming that each can be copied and consumed by many people at negligible additional cost.

\Cref{fig:feedback-loop} contains a high-level overview of the two-sided attention market (detail in \Cref{sec:model}). We take an optimization perspective
to examine the intertwined dynamics 
of user-platform and creator-platform interactions.
In our model, users respond to the {\it visibility}, {\it quality} and {\it popularity} of items. Aggregating user choices yields {\it popularity} signal across all items.
Creators are rewarded proportional to the {\it popularity} of items they create, and respond with a choice in item {\it quality} sensitive to cost of production. Platforms control item {\it visibility}, while relaying {\it quality} and {\it popularity} signals to both users and creators.
Prior research largely focused on one side of the market, such as user-platform side~\cite{zhu2023stability} or creator-platform side~\cite{benporat2018game,yao2023monotone,jagadeesan2023supplysideeqm}.
Our user model follows a trial-offer market dynamic~\cite{krumme2012quantifying} inspired by the Music Lab experiment~\cite{salganik2006experimental}. 
Content quality in this work assumes a scalar value; we leave vector-valued quality, such as those produced by personalisation and recommender systems~\cite{mcauley2022personalized}, as future work. 
Our creator reward vs cost model is inspired by \citet{GhoshHummel2014} except assuming bounded total reward. It affords the most general form of production cost, rather than assuming a particular functional form such as~\cite{jagadeesan2023supplysideeqm,yao2023-topK}. 
Game-theoretic analysis~\cite{benporat2018game,yao2023-topK} have provided useful insights including the qualities of equilibria and their convergence dynamics under no-regret learning. 
But our quest for an overall optimization problem provides additional benefits on predicting the long-term dynamics,
which in turn could inform the design of effective policy interventions that promote a healthy ecosystem, reduce the spread of undesirable content, and maintain fair income distributions among creators.

\begin{figure}[th!]
\includegraphics[width=1.0\linewidth]{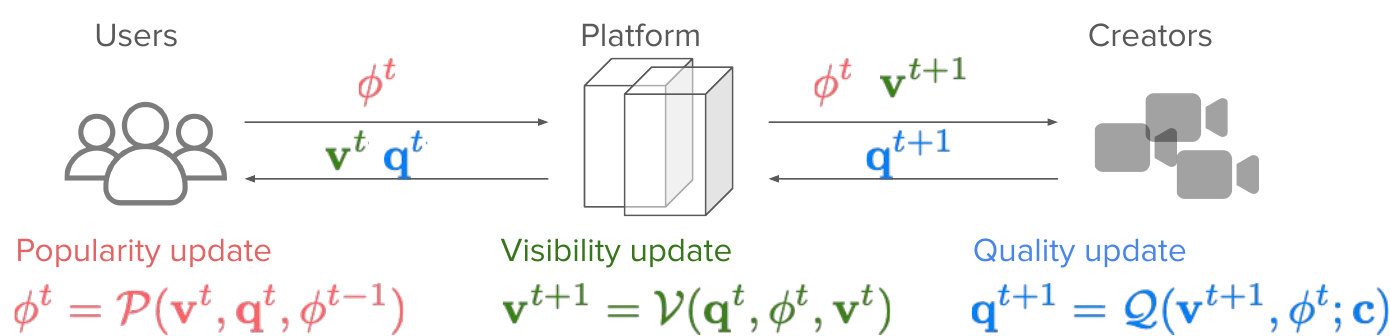}
\caption{An overview of two-sided attention market. Users chose content based on content {\it visibility} and {\it quality} to produce {\it popularity} signals. The platform updates {\it visibility} and pass along {\it popularity} to creators, which in turn drives creators to update content {\it quality}. Notations correspond to those in \Cref{sec:model}.}
\label{fig:feedback-loop}\vspace*{-0.1in}
\end{figure}

Our first result is uncovering a potential function underlying the two-sided attention market (\Cref{eq: interpretation-of-potential-0}), which is a combination of expected log-utility, cost of production, entropy, and the alignment between attention and visibility. This potential function is similar in spirit to those found in a potential game in that it expresses a global incentive for all parties to take action. 
We identify a transformation to express the function in a lower-dimensional space of user actions, rather than the joint action space of all players (\Cref{subsec: results-equiv-to-MD}).
The second result shows that
user and creator update dynamics correspond to
mirror descent, a fundamental optimization algorithm, on the potential function (\Cref{thm: dynamic-equiv-to-MD}).
Finally, we show that when the platform vary its
ranking strategies combining quality and popularity, there is a family of potential functions on which mirror descent (with momentum) still describes the corresponding response dynamic by users and creators (\Cref{subsec: result-recommendation-policy}).
To understand market dynamics governed by the non-convex potential functions, we complement the existing mirror descent toolbox by providing local convergence result (\Cref{thm: mirror-descent-convergence-non-convex-detailed}), which may be of independent interest.
\begin{figure*}[bt]
  \centering
 \includegraphics[width=1.\textwidth]{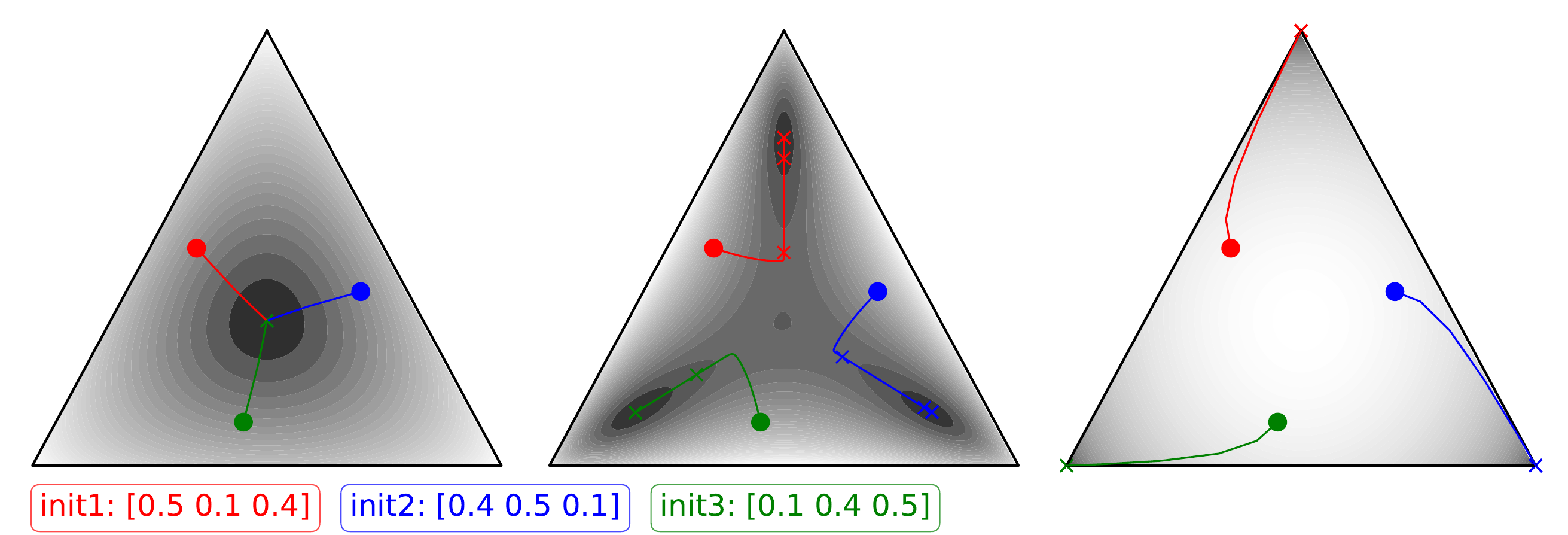}%
    \caption{Illustration of three instances of two-sided market dynamics. 
    Each triangle represents the projection of the probability simplex of a 3-item market in $\rr^3$. The contours within each simplex represent the values of the potential function of that market, where darker colours indicate smaller function values. The trajectories demonstrate the two-sided market dynamics with three different initialisations, where ``$\bullet$'' mark the initial points.
    }
    \label{fig:teaser}
\end{figure*}

The potential function is non-convex in general, meaning the market dynamics can reach different local minimisers, which correspond to equilibria with different properties. \Cref{fig:teaser} demonstrates the market dynamics under three instances of the markets. The contour plots represent the value of the corresponding potential functions defined in \Cref{eq: potential-function-0}. From left to right, the social influence factors $r = 0.1, 0.38$ and $0.5$. Similar landscapes of potential functions could also be equivalently generated by changing the costs of production $\bbc$ or implementing appropriate recommendation policies (cf. \Cref{subsec: rs-policy}). As shown in \Cref{fig:teaser}, though initialised the same, under different market instances, the market dynamics may converge to a unique minimiser, settle at different local minimisers, or move toward the boundary.

\subsubsection*{\textbf{Related Work.}} The attention markets of online media have garnered significant interest in recent years. Given the competitive nature among content creators, numerous game-theoretic analyses were performed on the creator-platform side of the markets, by considering strategy spaces such as content type, quality level, and clickbait use. Many studies have focused on Nash equilibria, coarse correlated equilibria and Shapley values of the games, and analyzed their efficiency (e.g., price of anarchy), fairness, and overall content quality.~\cite{immorlica2024clickbait,GhoshHummel2014,hron2023creatorincentive,jagadeesan2023supplysideeqm,yao2023-topK,yao2024userwelfareopt} Some of these games have been shown to be potential games, where best-response, better-response or no-regret learning dynamics converge to equilibria.~\cite{benbasat2017game,benporat2018game,benporat2020betterresponse,benporat2019convergence,yao2023monotone}

\citet{maldonado2018popularity} used a multinomial logit model to capture the interactions between recommender systems and users, and showed that the resultant dynamics converge to unique equilibrium. \citet{zhu2023stability} further demonstrated that these dynamics are equivalent to mirror descent on a convex function, which can be interpreted as an interpolation between market efficiency and diversity. This suggests that the market is implicitly performing meaningful optimization, echoing the famous \emph{invisible hand} insight of Adam \citet{smith1776wealth}. Indeed, a growing body of research at the intersection of computer science and economics has uncovered connections between natural market dynamics (e.g., proportional response, t\^atonnement) and optimization processes, strengthening the invisible hand insight.~\cite{cheung2018dynamics,birnbaum2011distributed,BranzeiDevanurRabani2021,CCD2020,gao2020first,kolumbus2023asynchronous,BranzeiMehtaNisan2018NeurIPS,CheungCole2018async} Our results provide a strong generalization of \cite{zhu2023stability}, by integrating dynamics on both sides of the markets, and identifying potential functions that capture the dynamics.

While some research has studied the dynamics on either side of attention markets, the study of full two-sided market dynamics is still in infancy.~\citet{dean2024recommender} proposed a framework for dynamical systems in two-sided attention markets, where our model fits.
They provided a comprehensive overview of key modelling aspects, and highlighted an important research goal in this area is to understand
the role of recommender systems in shifting viewer preferences~\cite{liu2023field} and shaping the content landscape~\cite{meyerson2012youtube}.
To the best of our knowledge, our work is the first to present a formal theoretical analysis of two-sided market dynamics.

This work substantially generalises the stability framework of attention markets developed by \cite{zhu2023stability}. Their analysis can be recovered as a special case of our model: it corresponds to creators repeatedly producing content of fixed quality without incurring production costs, under a platform that deploys a constant recommender system policy. Motivated by recent supply-side perspectives \cite{jagadeesan2023supplysideeqm,yao2023-topK,yao2024userwelfareopt,hron2023creatorincentive}, we endow creators with explicit utility functions and allow for the implementation of a range of recommendation policies. These additions are timely but analytically consequential: the resulting potential function is generally nonconvex, breaking the standard convex-optimisation correspondence used to study market dynamics—e.g., \cite{zhu2023stability} for attention markets and \cite{cheung2018dynamics,CCD2020,birnbaum2011distributed} for Fisher markets. Nonconvexity, in turn, yields local convergence behaviour and makes stability sensitive to step sizes, initialisations, and local landscape features determined by creators’ costs and the recommendation policy, a dependence we characterise precisely via our customised potential in \Cref{thm: customised-potential-function}.

\section{Model: Two-sided Attention Market}\label{sec:model}

A two-sided attention market comprises of three groups of agents: a set of content creators denoted by $\creators$, users,
and the platform's recommender system (RS) mediating in between. The market dynamics proceeds by \emph{epochs} $t=0,1,2,\ldots$.
In each epoch $t$, 
\begin{itemize}
\item Each creator $j\in \creators$ posts a media item of quality level $q_j^t\in [0,1]$ to the platform.
As noted in~\cite{GhoshHummel2014}, the quality level can be interpreted as the probability that the media item is purchased by a user after trial.
We will use the index $j$ to refer to either the creator or their media item depending on the context.
\item The RS determines a visibility parameter $v_j^t \in [0,1]$ for each creator $j$. The visibility parameters are used to compute the probabilistic recommendations of media items to users.
\item Users visit the platform and receive recommendations from the RS. Each user chooses a media item for trial.
If the user likes the item, they will purchase it. In the \textsc{MusicLab} setting, the purchase is solely determined by the quality level of the item.~\cite{maldonado2018popularity}
The purchases change popularities of items, and also revenues of creators. Let $\phi_j^t$ denote the fraction of attention allocated to creator $j$ at the end of this epoch, henceforth called {\em popularity}.
\item At the end of the epoch, the creators and the RS observe the popularities and revenues, and then renew their quality levels and visibility parameters respectively.
\end{itemize}
In the rest of this section, we discuss the motivations and mathematical formulations of the updates by creators, RS and users. \Cref{app:update-summary} provides a concise summary of the mathematical specifications of the dynamical systems, which you may find it useful when reading our technical discussions in later sections.

\subsection{General Form of Dynamical Systems in Two-Sided Attention Markets}

Each creator $j$ incurs a cost $c_j(q_j^t)$ when creating a media item of quality level $q_j^t$.
The cost function $c_j: [0,1] \rightarrow \rr^+$ is increasing, with $c_j(0) = 0$. The creator's objective is to maximize profit, i.e., revenue minus cost.

At each epoch, the popularities, visibility parameters and quality levels  %
are updated by the users, RS and creators respectively as follows: %
\begin{itemize}
    \item \textbf{Popularity update}: In each epoch, the users visit the platform and interact with the items which redistribute the popularity signals. Let $\calP$ denote the operator describing the user-item interactions:
    \begin{equation}
        \bphi^{t} = \calP(\bbv^{t}, \bbq^{t}, \bphi^{t-1}). \label{eq: popularity-update-rule-general}
    \end{equation}

    \item \textbf{Visibility update}: The RS adjusts the visibilities of the items based on the items' popularities and qualities. Let $\calV$ denote the operator describing platform's recommendation strategy:
    \begin{equation}
         \bbv^{t+1} = \calV(\bbq^{t}, \bphi^{t}, \bbv^{t}). \label{eq: visibility-update-rule-general}
    \end{equation}

    \item \textbf{Quality update}: The creators update the qualities of their items by considering their potential profit given the visibilities, popularities and the costs of production. Let $\calQ$ denote the operator describing the update:
    \begin{equation}
        \bbq^{t+1} = \calQ(\bbv^{t+1},\bphi^{t}; \bbc), \label{eq: quality-update-rule-general}
    \end{equation}
    where cost function $\bbc := \{c_j:~ j\in \creators\}$ does not change over time.  
\end{itemize}
Please refer to \Cref{fig:feedback-loop} for a visualization of the feedback loop in the dynamical systems described above.
In the subsequent subsections, we introduce the specific realizations of these updates examined in this paper. In \Cref{subsec: trial-offer-market}, we present the trial-offer market dynamics --- a well-established model of user-platform interactions that defines $\calP$~\cite{maldonado2018popularity,zhu2023stability}. In \Cref{subsec: producer-best-respond}, we outline the creators’ best-response behaviours, which yields a realization for 
$\calQ$. In \Cref{subsec: rs-policy}, we discuss potential strategies for $\calV$ that the platform may employ.

\subsection{User-Platform Interaction: Trial-Offer Market with Social Influence} \label{subsec: trial-offer-market}
Let $t\in \nn$ denote the index of a epoch, and $\tau \in \nn$ denote the time-step within a epoch $t$.
We consider two types of user behaviours, which correspond to different responding speeds of the creators. 

\paragraph{Stochastic Trial-Offer Market Dynamics} At each time $\tau \in \nn$, a user comes to the platform and tries an item, and then they decide to purchase the item or not. Let $d_j^{\tau, t}$ denote the number of purchases of item $j$ up
to time $\tau$ within epoch $t$. To ensure that all items have a positive probability to be tried initially, we assume that at the start of each epoch, $d_j^{0, t} \geq 1$ for every item $j$. %
The popularity of item $j$ is the fraction of all purchases which happens to item $j$:
\[
\phi_{j}^{\tau, t} := \frac{d_j^{\tau, t}}{\sum_i d_i^{\tau, t}}.
\]
We note that the possible popularity vectors form a simplex, denoted by $\Delta$:
\begin{equation}
\Delta = \left\{ \boldsymbol{\phi} \in \rr:~ \sum_{j=1}^{|\creators|} \phi_j = 1, \phi_j \geq 0~\text{for any item }j \right\}. \label{eq: def-simplex}
\end{equation}
The probability that an item is tried by the user at time-step $\tau$ is modelled as a multinomial logit:%
\begin{equation}
       \prob{\text{Item $j$ is chosen by trial}}:= \frac{v_j^t (\phi^{\tau, t}_j)^r}{\sum_i v_i^t (\phi^{\tau, t}_i)^r},
\end{equation}
where $r \geq  0$ denotes the significance of social influence. A larger value of $r$ represents stronger conformity of the users. 
After choosing the item for trial, the user will decide whether to purchase this item according to its quality. Within each epoch,
the probability of purchasing item $j$ given it is chosen for trial is given by the quality level $q_j^t \in [0,1]$.
By direct computation, one can derive the probability that the next successful purchase happens on item $j$.

\begin{lemma}[\cite{maldonado2018popularity} Lemma 3.1]
    Within epoch $t$, the probability that the next purchase is the product $j$ given the popularity vector $\bphi$ is given by
    \[
    p_j(\bphi) = \frac{v_j^t q_j^t (\phi_{j})^r}{\sum_i v_i^t q_i^t (\phi_{i})^r}.
    \]

\end{lemma}
With the above lemma, one could observe that the stable point of the stochastic process is the status that probability of the next success purchase is exactly the current market share, which gives the following definition.
\begin{definition}
    For any trial-offer market, we say a market share $\bphi$ is a trial-offer market equilibrium (TOME) if 
    \(
     p_j(\phi) = \phi_j, ~~~\text{for all }j\in \creators.
    \)
     We say $\bphi$ is an interior TOME if it is a TOME with $\phi_j > 0$ for every item $j$.

\end{definition}

One can directly calculate the interior TOME for any epoch $t$, which is
\begin{equation}
 \phi_j^t := \frac{(q_j^t v_j^t)^{\frac{1}{1-r}}}{\sum_i (q_i^t v_i^t)^{\frac{1}{1-r}}}. \label{eq: interior-TOME}
\end{equation}
It was shown that the stochastic trial-offer market will converge to TOME almost surely. The result is summarised in the following theorem.
\begin{theorem}[\cite{maldonado2018popularity} Theorem 5.1, Theorem 5.3, \cite{van2016aligning} Theorem 3]
    The stochastic Trial-Offer market dynamics converge to TOME almost surely. 
    \begin{itemize}
        \item If $r\in (0,1)$, the limit point is the interior TOME given by \Cref{eq: interior-TOME}.
        \item If $r = 1$, the limit point is the $j^*$-th standard basis vector of $\rr^{|\creators|}$ with $q^{t}_{j^*} \geq q_i^{t}$ for all $i \in \creators$.
    \end{itemize}
    \label{thm: convergence-stochastic-TO-market}
\end{theorem}
With the above theorem, it is reasonable to view %
that each epoch corresponds to the whole process in which the stochastic dynamics converge to the TOME. In other words, with the stochastic trial-offer dynamic, the creators and the platform will enter the next epoch once after the stochastic process converges to the equilibrium. This gives the definition of the Equilibrium Response (ER) dynamic.
\begin{definition}[Equilibrium Response (ER) Dynamic] \label{def:ER}
    Given the set of items $\creators$, the two-sided attention market dynamics described by \Cref{eq: popularity-update-rule-general,eq: quality-update-rule-general,eq: visibility-update-rule-general} with the popularity update operator
    \[
     \left[\calP_{\mathrm{ER}}(\bbv^{t}, \bbq^{t}, \bphi^{t-1})\right]_j := \frac{(q_j^t v_j^t)^{\frac{1}{1-r}}}{\sum_i (q_i^t v_i^t)^{\frac{1}{1-r}}},~~~~~\text{for every entry }j\in\creators,
    \]
    is called the Equilibrium Response (ER) dynamic.
\end{definition}

\paragraph{Deterministic Trial-Offer Market Dynamics} We introduce a deterministic analog of the stochastic dynamic introduced above. Here, we are able to adjust the definition of epochs. Instead of waiting until the stochastic dynamic converges, the creators and the RS might adjust their strategies more frequently.  

Suppose that there is a certain number of users entering the platform \emph{simultaneously} in each epoch. We assume the users are non-atomic and every user will choose exactly one item for trial. At epoch $t+1$, the probability that item $j$ is tried by any user is
\[
 \frac{v_j^{t+1} (\phi_j^t)^r}{\sum_i v_i^{t+1} (\phi_i^t)^r}.
\]
Again, the item will be purchased with probabiliy $q_j^t \in [0,1]$. Instead of accumulating the market share as the atomic-user case, here the popularities are updated as the fraction of purchases happened to an item over the epoch. By direct computation, one can get that the market share should be updated\footnote{
Note that the popularity signal is only updated once an epoch is finished. Also, note that the convergence result of the stochastic dynamic also does not put any assumptions on the initial conditions other than $\bphi^{0, t} > 0$. Therefore, it is acceptable that the popularity signal does not carry through epochs for both deterministic and stochastic case.} as
\[
 \phi^{t}_j = \frac{q_j^t v_j^t (\phi_j^{t-1})^r }{\sum_i q_i^t v_i^t (\phi_i^{t-1})^r }.
\]
This gives the definition of the proportional response dynamic.
\begin{definition}[Proportional Response (PR) Dynamic] \label{def:PR}
    Given the set of items $\creators$, the two-sided attention market dynamics described by \Cref{eq: popularity-update-rule-general,eq: quality-update-rule-general,eq: visibility-update-rule-general} with the popularity update operator
    \[
     \left[\calP_{\mathrm{PR}}(\bbv^{t}, \bbq^{t}, \bphi^{t-1})\right]_j := \frac{q_j^t v_j^t (\phi_j^{t-1})^r }{\sum_i q_i^t v_i^t (\phi_i^{t-1})^r },~~~~~\text{for every item }j\in\creators,
    \]
    is called the Proportional Response (PR) dynamic.
\end{definition}

Intuitively, the deterministic dynamic corresponds to the assumption that the users are non-atomic whereas the stochastic dynamics assumes the users are discrete. Indeed, these two types of dynamics are analogous to each other in the sense that they have common limit points. 

\begin{theorem}[\cite{zhu2023stability} Theorem 3.7]
    With $r \leq 1$, the stochastic trial-offer market dynamics and the deterministic trial-offer market dynamics converge to the same TOME.
\end{theorem}
    
The trial-offer market model can be traced to the large-scale experimental study by \citet{salganik2006experimental} towards understanding the inequality and predictability of online cultural market. Based on the experimental, it turns out that the user behaviour could be cast into a multinomial logit model \cite{krumme2012quantifying}. And the convergence behaviour of the dynamics are further analysed by \citet{ceyhan2011social, van2016aligning, maldonado2018popularity, zhu2023stability}.

\subsection{Platform-Creator Interaction} \label{subsec: producer-best-respond}
At the end of each epoch, each creator updates their item quality level by maximising the estimated income they can earn in the next epoch. We assume the estimated income is given by the expected purchases given by the status at the end of the epoch. Denote the estimated income of creator $j$ at the end of epoch $t$ by $\mathrm{inc}^t_j$. Given any $q_j$, the estimated income is
\begin{align*}
&~~\expect{\mathrm{inc}^t_j ~|~ \bbv^{t+1}, \bphi^t, q_j} \\
&=\prob{\text{Item }j\text{ is purchased}~|~\text{Item }j\text{ is chosen for trial}} \cdot \prob{\text{Item }j\text{ is chosen for trial}}\\
&= q_j \cdot \frac{v_j^{t+1} (\phi_j^t)^r}{\sum_i v_i^{t+1} (\phi_i^t)^r}.
\end{align*}
Given the cost function of creator $j$, their utility function is
\[
 u_j^t(q_j):=\expect{\mathrm{inc}^t_j ~|~ \bbv^{t+1}, \bphi^t, q_j} - c_j(q_j) = q_j \cdot \frac{v_j^{t+1} (\phi_j^t)^r}{\sum_i v_i^{t+1} (\phi_i^t)^r} - c_j(q_j).
\]
The creator will choose the quality level of their next item as the utility maximiser, i.e.,
\[
 q_j^{t+1} = \left[\calQ(\bbv^{t+1},\bphi^{t}; \bbc)\right]_j = \argmax_{q_j} u^t_j(q_j).
\]
If the cost function $c_j$ is strictly convex, then the above update is well-defined since the utility function $u_j^t$ is strictly concave. 
By simple calculus, we obtain a closed-form formula of the quality level update at the end of each epoch.
\begin{lemma}
    Suppose that $c_j$ is strictly convex and continuously differentiable, $c_j'(0) = 0$ and $c_j'(1)\ge 1$ for every $j\in \calC$. Given the quality $\bphi^t$ and visibility factors $\bbv^{t+1}$. Let $\zeta_j := (c_j')^{-1}$ denote the inverse function of the derivative of the cost function of creator $j$. Then the quality update can be rewritten as
    \[
    q^{t+1}_j = \left[\calQ_{\mathrm{BR}}(\bbv^{t+1},\bphi^{t}, \bbc)\right]_j := \zeta_j\left(\frac{v_j^{t+1} (\phi_j^t)^r}{\sum_i v_i^{t+1} (\phi_i^t)^r} \right).
    \] \label{lem: response-of-producers}
\end{lemma}

\subsection{Recommendation Policies of the Platform} \label{subsec: rs-policy}
As the quality levels and popularities of the items evolve, the platform may implement a RS to intervene the market dynamics.
We consider the following recommendation policies and their effects on the dynamic:
\begin{itemize}
    \item \textbf{Popularity ranking}: The visibilities are set to be proportional to the popularity of items of the last epoch:
    \begin{equation}
     \left[\calV_{\mathrm{pop}}(\bbq^{t}, \bphi^t, \bbv^t)\right]_j:= \frac{\mu_j(\phi^t_j)^\beta}{\sum_i \mu_i(\phi^t_i)^\beta}, ~~~~~\text{for every entry }j\in\creators, \label{eq: popularity-ranking-policy}
    \end{equation}
    where $\beta \in \rr$ is a parameter reflecting the strength of recommendation policy and $\boldsymbol{\mu} > 0$ is a constant.
    \item \textbf{Quality ranking}: The platform sets the visibility factors proportional to the qualities of the items at the current epoch:
    \[
    \left[\calV_{\mathrm{qual}}(\bbq^{t}, \bphi^t, \bbv^t)\right]_j:= \frac{\mu_j(q^{t}_j)^\alpha}{\sum_i \mu_i(q^{t}_i)^\alpha}, ~~~~~\text{for every entry }j\in\creators,
    \]

    \item \textbf{Mixed ranking}: The platform sets the visibility factors proportional to the product of qualities and popularities of the items at the current epoch:
    \[
    \left[\calV_{\mathrm{mix}}(\bbq^{t}, \bphi^t, \bbv^t)\right]_j:= \frac{\mu_j(q^{t}_j)^\alpha (\phi^t_j)^\beta}{\sum_i \mu_i (q^{t}_i)^\alpha (\phi^t_i)^\beta}, ~~~~~\text{for every entry }j\in\calI,
    \]

    \item \textbf{Constant}: The visibilities may remain constant throughout all epochs,
    \[
    \calV_{0}(\bbq^{t+1}, \bphi^t, \bbv^t) := \bbv > 0,
    \]
    where $\bbv$ is some constant vector with $\sum_j v_j = 1$.
\end{itemize}
The visibilities are normalised such that they sum to $1$. The parameters $\beta, \alpha$ can potentially be negative.
This can happen if the platform wants to demote items with high visibilities or popularities due to fairness concerns.
To actually implement quality ranking, the platform needs to learn the quality levels by soliciting sufficient trial data or by implementing a bandit algorithm.

\section{Results} \label{sec:results}
We characterize the convergence behaviour of the dynamics in two-sided attention markets by casting them as optimization processes on carefully designed potential functions.
Specifically, we will show that the dynamics are equivalent to \emph{mirror descent} algorithm.
The equivalence allows us to understand these dynamics using the existing knowledge in research of optimization.

Since the potential functions are not always convex, it may possess multiple local minima.
Depending on the initial conditions, step-sizes and other factors, each local minima can become a possible long-term outcome.
In general, it is very difficult to describe cleanly how those factors determine the outcome.
Nevertheless, \emph{local convergence} --- convergence to a local minima if the initial condition is sufficiently close to the minima --- can be demonstrated.

In \Cref{subsec: mirror-descent-results}, we discuss the background knowledge and new results about mirror descent that are essential for comprehending our results about attention market dynamics. In \Cref{subsec: results-equiv-to-MD,subsec: result-convergence-two-sided-dynamic}, we present the potential function when the recommender system implements constant policy, show that market dynamics are equivalent to mirror descent on the potential function, and discuss the convergence properties of the dynamics. We extend these results to other recommendation strategies in \Cref{subsec: result-recommendation-policy}.

\subsection{Mirror Descent Algorithm, Global \& Local Convergence}\label{subsec: mirror-descent-results} %
Consider a general constrained %
optimization problem of minimizing a smooth %
function $f(\bbx)$, subject to the constraint $\bbx\in X$ for some compact and convex set $X$.
\begin{definition}
    Let $X$ be a compact and convex set, and let $h$ be a differentiable convex function on $X$. The Bregman divergence w.r.t. $h$, denoted by $d_h$,
    is defined as 
    \[
    d_h(\bbx,\bby) = h(\bbx) - h(\bby) - \inner{\nabla h(\bby)}{\bbx- \bby},
    \]
    for any $\bbx \in X$ and $\bby \in \mathrm{rint}(X)$, where $\mathrm{rint}(X)$ denotes the relative interior of the set $X$.
\end{definition}

The widely used Kullback–Leibler (KL) divergence is a special case of Bregman divergence, generated by the function $h(\bbx) = \sum_j x_j\log x_j - x_j$.

Given a Bregman divergence $d_h$, the corresponding mirror descent update rule is
\begin{equation}
 \bbx^{t+1} = \argmin_{\bbx \in X} \left\{ \inner{\nabla f(\bbx^t)}{\bbx - \bbx^{t}} + \frac{1}{\eta} \cdot d_h(\bbx, \bbx^t)\right\}, \label{eq: mirror-descent-update-rule}
\end{equation}
where $\eta$ is considered as the step-size of the update rule. %
In particular, if $d_h$ is the KL divergence and $X = \Delta$, then the mirror descent update becomes: for each $j$, %
\begin{equation}
x_j^{t+1} = \frac{x_j^t \exp(-\eta \nabla_j f(\bbx^t))}{\sum_i x_i^t \exp(-\eta \nabla_i f(\bbx^t))}~.\label{eq: mirror-descent-KL}
\end{equation}

In optimization, it is well known that if the function $f$ is convex and satisfies some regularity conditions, then mirror descent update \eqref{eq: mirror-descent-update-rule} converges to the (essentially) unique minima of $f$. However, if $f$ is not convex, the story is far more involved. 

\subsubsection*{\textbf{Convex Objective Function and Global Convergence}}
We define the notion of Bregman-smoothness below, which is a necessary condition for the mirror descent dynamic to converge.
\begin{definition}
    The function $f$ is $L$-Bregman-smooth with respect to Bregman divergence $d_h$, if for any $\bby\in \mathrm{rint}X$ and $\bbx \in X$,
    \[
    f(\bbx) \leq  f(\bby) + \inner{\nabla f(\bby)}{\bbx - \bby } + L \cdot d_h(\bbx, \bby).
    \]
    Moreover, it is $L$-Bregman-convex if the function is both $L$-Bregman-smooth and convex.
\end{definition}

\begin{theorem}[{\cite[Theorem 3]{birnbaum2011distributed}}]
Suppose $f$ is an $L$-Bregman-convex function with respect to $d_h$, let $\bbx^T$ be the point reached by running mirror descent \Cref{eq: mirror-descent-update-rule}. Then with the learning rate $\eta \leq \frac{1}{L}$,
\[
f(\bbx^T) - f(\bbx^\star) \leq \frac{d_h(\bbx^\star, \bbx^0)}{\eta\cdot T}. \label{thm: mirror-descent-convergence-convex}
\]
\end{theorem}

\subsubsection*{\textbf{Non-convex Objective Function and Local Convergence}} Now, consider the case where $\Phi_0$ is not convex everywhere.
The intuition is that if the dynamics are initialized near a local minimiser, they may still converge to it at the same convergence rate $\mathcal{O}\left(\frac{1}{T}\right)$. %
While this seems intuitive, formally establishing it requires non-trivial mathematical reasoning.
To keep things simple, we specifically consider the case where the Bregman divergence is the KL-divergence, and the optimization problem is constrained to the simplex $\Delta$, which is the setting we need for analyzing the two-sided market dynamics. 

First, we define the regularity condition of the local minimizer as follows. We first perform a change of variable to reduce the dimensionality. Suppose $f: \Delta  \rightarrow \rr$, where $\Delta \subset \rr^n$.
For any $\widehat{x} \in \rr^{n-1}$, we set $\bbx = \left[\widehat{\bbx}, 1- \sum_{j=1}^{n-1} \widehat{x}_j\right]$. Then define $\widehat{f}: \rr^{n-1} \rightarrow \rr$ such that $\widehat{f}(\widehat{\bbx}) = f(\bbx)$. %
\begin{definition}
    A strict local minimiser $\bbx$ of $f: \Delta \subset \rr^n \rightarrow \rr$ is regular if the corresponding $\widehat{f}$ and $\widehat{x}$ defined above satisfies $\nabla^2 \widehat{f}(\widehat{x}) \succ 0$.
\end{definition}
We note that the strict local minimiser of $f$ within $\mathrm{rint}\Delta$ is also a strict local minimiser of $\widehat{f}$. Therefore, it must hold that $\nabla^2 \widehat{f} \succeq 0$. The above condition further regularise the local minimiser such that the objective is not completely \emph{flat} around the local minimiser. 

\boxlem{
For some open set $\Omega \subset \rr^n$, suppose $\bbx^\star\in \mathrm{rint}\Delta$ is a strict regular local minimiser for some smooth function $f$ over $\Omega \cap \Delta$. Suppose $f$ has finitely many local minimisers in $\mathrm{rint}\Delta$, there exists a open set $\Omega' \subset \Omega$ such that,
\begin{align}
f(\bbx^\star) -  f(\bby) - \inner{\nabla f(\bby)}{\bbx^\star -\bby} \geq  0~, \label{eq-lem: local-convexity-line-1}\\
\inner{\nabla f(\bby)}{\bby - \bbx^\star} -\gamma \Vert \bby - \bbx^\star \Vert_2^2 \geq 0~, \label{eq-lem: local-convexity-line-2}\\
\left\Vert \nabla f(\bby) - \nabla f(\bbx^\star)  \right\Vert_2 -\gamma\left\Vert \bby - \bbx^\star \right\Vert_2\geq  0~.\label{eq-lem: local-convexity-line-3}
\end{align}
for any $\bby \in \Omega' \cap \Delta$ and some $\gamma > 0$. \label{lem: local-convexity-around-strict-minimiser}
}
Since $f$ is twice differentiable, then $\nabla f$ is $\kappa$-Lipschitz continuous for some $0<\kappa<\infty$ within some sufficiently small neighbourhood $\Omega$. We can formulate the ``attractiveness'' of the minimiser as follows.
\begin{definition}
    The neighbourhood $\Omega(\bbx^\star)$ of a regular stirct local minimiser is $(\gamma, \kappa)$-smooth  if \Cref{eq-lem: local-convexity-line-1,eq-lem: local-convexity-line-2,eq-lem: local-convexity-line-3} is satisfied with parameter $\gamma$ and the $\nabla f$ is $\kappa$-Lipschitz continuous over $\Omega(\bbx^\star)$. \label{def: smooth-neighbourhood}
\end{definition}

Finally, we are ready to state the main result about local convergence for non-convex functions.
\begin{theorem}
Let $\bbx^\star$ be an interior regular strict local minimiser of $f: \Delta \subset \rr^n \rightarrow \rr$. Assume $f$ is twice differentiable, $L$-Bregman smooth and has finitely many local minimisers. Suppose a minimiser $\bbx^\star$ is equipped with a $(\gamma,\kappa)$-smooth neighbourhood $\Omega$ and $\bbx^0\in \Omega$,
and the mirror descent update \eqref{eq: mirror-descent-KL} with learning rate $\eta \leq \min\left\{\frac{2\gamma}{\kappa^2}, \frac{1}{L}\right\}$ is used. Then 
\[
 f(\bbx^T) - f(\bbx^\star) \leq \frac{d_h(\bbx^\star, \bbx^0)}{\eta\cdot T}. \label{thm: mirror-descent-convergence-non-convex-detailed}
\]
\end{theorem}

\subsection{Two-Sided Attention Market Dynamic under Constant RS Policy}\label{subsec: results-equiv-to-MD}
In this subsection, we consider the cases where the RS's policy is constant, i.e., $\calV_{0}(\bbq^{t+1}, \bphi^t, \bbv^t) = \bbv$ for some $\bbv > 0$.
Our potential function takes a vector $\bbs^t$ as input, where
\[
s_j^t := \prob{\text{Item $j$ is chosen by trial at the end of epoch $t$}} = \frac{v_j (\phi^{ t}_j)^r}{\sum_i v_i (\phi^{ t}_i)^r}~.
\]
We note that when $\bbs^t \in \Delta$ is given, it is easy to derive the corresponding $\bphi^t\in \Delta$ uniquely.
It turns out that no matter the users' update is ER or PR,
$\bbs^{t+1}$ admits an iterative formula in term of $\bbs^t$, which turns out to be a mirror descent update on a designed potential function.

\boxlem{Assume $r<1$. If the users' update is the equilibrium response dynamic in \Cref{def:ER}, we have
\begin{equation}
s_j^{t+1} = \frac{v_j^{\frac{1}{1-r}}\left(\zeta_j(s_j^t)\right)^{\frac{r}{1-r}}}{\sum_{i}v_i^{\frac{1}{1-r}}\left(\zeta_i(s_i^t)\right)^{\frac{r}{1-r}}}. \label{eq: ER-dynamic-close-form}
\end{equation} 
Alternatively, if the users' update is the proportional response dynamic in \Cref{def:PR}, we have
\begin{equation}
s_j^{t+1} = \frac{v_j (s_j^t\zeta_j(s_j^t))^r}{\sum_{i\in \calC}v_i (s_i^t\zeta_i(s_i^t))^r}. \label{eq: PR-dynamic-close-form}
\end{equation}
\label{lem: dynamics-formulation}
}

\Cref{lem: dynamics-formulation} indicates that both ER and PR dynamics correspond to redistribution and renormalization processes in proportional to the product of $v_j$, $s_j^t$ and $q_j^{t+1} = \zeta_j(s_j^t)$ up to some powers that depend on $r$. Optimization researchers might have a feeling of d\'ej\`a vu.

\begin{theorem}
Define the potential function $\Phi_0$ as
\begin{equation}
\Phi_0(\bbs) := -\left(\sum_{j} s_j \log v_j + r \cdot\int_{0}^{s_j} \log \zeta_j(z) ~\mathrm{d}z + (r-1)s_j\log s_j  \right). \label{eq: potential-function-0}
\end{equation}
    \begin{itemize}
    \item The ER dynamic \eqref{eq: ER-dynamic-close-form} is equivalent to mirror descent on $\Phi_0$ with KL divergence on $\Delta$, with the corresponding learning rate $\eta =\frac{1}{1-r}$,
    \item The PR dynamics stated in \eqref{eq: PR-dynamic-close-form} is equivalent to mirror descent on $\Phi_0$ with KL divergence on $\Delta$, with the corresponding learning rate $\eta =1$.

\end{itemize}
\label{thm: dynamic-equiv-to-MD}
\end{theorem} 

\Cref{thm: dynamic-equiv-to-MD} indicates that ER and PR dynamics are both equivalent to the mirror descent on the same potential function $\Phi_0$, but with different learning rates. With the assumption that $0<r<1$, we have $\frac{1}{1-r}> 1$, hence the ER dynamic has a larger corresponding learning rate. This echoes with our model assumptions that ER dynamic corresponds to response to accumulation of more users' behaviours, which means that each epoch should correspond to a \emph{larger step} in terms of mirror descent on the same potential function.

\subsubsection*{\textbf{An interpretation of the potential function}} Here, we argue that the potential function can be interpreted meaningfully. After a transformation of the potential function (details are given in \Cref{app-sec: transformation-of-potential-function}), it has the following decomposition:
\begin{equation}
\Phi_0\left(\bbs^t\right) = \underbrace{\mathrm{KL}(\bbs^t, \bbv)}_{\substack{\text{Attention alignment} \\ \text{to visibility}}}
 - r\cdot\underbrace{\sum_j s_j^t\log\left(q_j^{t+1}\right)}_{\text{Expected $\log$-quality}} ~+~  r\cdot\underbrace{H(\bbs^t)}_{\text{Entropy}} ~+~ r\cdot\underbrace{\sum_j\int_{0}^{q_j^{t+1}} \frac{c_j'(u)}{u} ~\mathrm{d}u}_{\text{Cost of production}}  , \label{eq: interpretation-of-potential-0}
\end{equation}
where the function $H(\bbs^t) = -\sum_j s_j^t\log s_j^t$ is the \emph{Shannon entropy} of the trial probability distribution. \Cref{eq: interpretation-of-potential-0} indicates that the dynamics can be viewed as a minimization process with trade-off of four aspects.
The first term is the KL divergence between the trial probability distribution $\bbs^t$ and the visibility distribution $\bbv$, which is preset by the platform. The market dynamics tend to minimise their gap such that the attention allocation of the users is approaching $\bbv$ desired by the platform. 
The second term can be viewed as the negation of the expected ``$log$-quality'', so the market dynamics tend to maximize this form of overall market quality/efficiency.

The third term is the Shannon entropy, which is maximized when its input is the uniform distribution, and it is minimized when the input is a degenerate distribution (i.e., with probability 1 on one of the items). Thus, it induces a tendency for the market dynamics to move away from uniform distribution and to \emph{polarize}.

Finally, the fourth term measures the total ``weighted'' cost of production in the market. When compared it with the plain cost of production $c_j(q_j) = \int_0^{q_j} c'(u) ~\mathrm{d}u$, the third term amplifies the cost when the quality level is low. It encourages the participation of content creators who are capable of producing low-quality items with relatively low costs. On the other hand, the market dynamics tend to block the creators who need significant effort to produce low-quality items, though they may have good potential to create high-quality contents with lower costs. This corresponds to the notion of ``barrier to entry'' \cite{mcafee2004whatIsBarrierToEntry}, where the market implicitly establishes a structural barrier to entry \cite{lutz2010perceptions} such that the creators need to first demonstrate the ability to enter the market by producing low-quality items in order to enter the market.

Notably, by increasing the strength of social signal $r$, all of the last three terms are more significant. It implies that the social signal is not only polarising the resulting attention distribution, but also improving the production efficiency through simultaneously reducing the input costs and improving the item qualities.

\subsection{Convergence Behaviours of the Dynamics} \label{subsec: result-convergence-two-sided-dynamic}
Clearly, any stationary points of the ER/PR dynamics must be the fixed point of the update rules \eqref{eq: ER-dynamic-close-form} and \eqref{eq: PR-dynamic-close-form} respectively.
By direct calculations, we notice that the fixed points of the two dynamics are exactly the same, which are the critical points of the Lagrangian of $\Phi_0$.
This motivates the following equilibrium definition.
\begin{definition}
    A point $\bbs^\star \in \Delta$ is called the \emph{two-sided trial-offer market equilibrium} (TSTOME) if it is a fixed point of the update \eqref{eq: ER-dynamic-close-form} or \eqref{eq: PR-dynamic-close-form}. Specifically, it satisfies
    \[
     s_j^{\star} = \frac{v_j (s_j^{\star}\zeta_j(s_j^{\star}))^{r}}{\sum_{i}v_i (s_i^{\star}\zeta_i(s_i^{\star}))^r}.
    \]
\end{definition}

\subsubsection*{\textbf{Convex Potential Function Cases}}
If $\Phi_0$ is a convex function, then it admits a (essentially) unique local minimizer, which corresponds to a TSTOME.
We explore conditions that guarantee $\Phi_0$ is convex.
Since the cost function $c_j$ is convex and $\zeta_j = (c')^{-1}$, $\zeta_j$ is an increasing function.
By direct computation of the Hessian of $\Phi_0$, we derive a sufficient and necessary condition that guarantees convexity of $\Phi_0$ at a given point.

\begin{lemma}
$\Phi_0$ is $(1-r)$-Bregman-smooth for if $\zeta_j$ is increasing for all $j\in \mathcal{C}$.
\label{lem: Bregman-smooth-phi0}
\end{lemma}

\begin{lemma}\label{lem: convexity-condition-Phi-0}
    $\Phi_0(\bbs)$ is convex at the point $\bbs$ if and only if 
    \begin{equation}
        rs_j\cdot\zeta_j'(s_j) + (r-1)\zeta_j(s_j) \leq 0, ~~~~~~~\forall j \in \calC. \label{eq: convexity-condition-phi-0}
    \end{equation} 
\end{lemma}
The convexity condition \Cref{eq: convexity-condition-phi-0} can be interpreted as follows. Consider the function $\eta_j(x) = (r-1)\log(x) + r\log(\zeta_j(x))$. Then $\eta_j'(x) =\frac{rx\cdot\zeta_j'(x) + (r-1)\zeta_j(x)}{x\zeta_j(x)} $. The sign of this value indicates different directions of monotonicity of $\eta_j(x)$. Suppose $\eta_j'(x) \leq 0$, for all $x\in (0,1)$, for some $x_1 \geq x_2 \in (0,1) $, one will have
\begin{equation}
\eta_j(x_1) - \eta_j(x_2)  = \log\left(\frac{x_1}{x_2}\right)^{r-1} + \log\left(\frac{\zeta_j(x_1)}{\zeta_j(x_2)}\right)^r = \log \left(\frac{x_1}{x_2}\right)^{r-1}\left(\frac{\zeta_j(x_1)}{\zeta_j(x_2)}\right)^r \leq 0.
\label{eq: concave-condition}
\end{equation}
This is equivalent to
\[
\frac{\zeta_j(x_1)}{\zeta_j(x_2)} \leq \left(\frac{x_1}{x_2}\right)^{\frac{1-r}{r}},
\]
which means that to guarantee the convexity of $\Phi_0$ (hence the \emph{global} stability of the dynamic), it requires the (multiplicative) growth rates of the response functions $\zeta_j$ not being too large. Equivalently, content creators ought not be motivated to drastically alter quality level in response to a slight increase in income.

\begin{theorem}[Global Convergence under convexity]
    Suppose that $r< 1$, and the condition \eqref{eq: convexity-condition-phi-0} is satisfied for all $\bbs \in \Delta$, then the ER dynamic satisfies 
    \begin{equation}
    \Phi_0\left(\bbs^T\right) - \Phi_0^\star  \leq  \frac{(1-r)\cdot\mathrm{KL}(\bbs^\star, \bbs^0)}{ T}, \label{eq: convex-convergence-rate-ER}
    \end{equation}
    where $\Phi^\star$ is the minimum value of $\Phi_0$. And the PR dynamic satisfies
    \begin{equation}
       \Phi_0\left(\bbs^T\right) - \Phi_0^\star  \leq  \frac{\mathrm{KL}(\bbs^\star, \bbs^0)}{T}.\label{eq: convex-convergence-rate-PR}
     \end{equation}
    Moreover, if \Cref{eq: convexity-condition-phi-0} is satisfied strictly, both the ER and PR dynamics converge to the TSTOME which minimises $\Phi_0$ globally. \label{thm: dynamic-converge-convex}
\end{theorem}
\begin{proof}
    By \Cref{lem: Bregman-smooth-phi0} and \Cref{lem: convexity-condition-Phi-0}, the potential function $\Phi_0$ is $(1-r)$-Bregman-convex. By \Cref{thm: mirror-descent-convergence-convex} and \Cref{thm: dynamic-equiv-to-MD}, the convergence results \Cref{eq: convex-convergence-rate-ER,eq: convex-convergence-rate-PR} follow. If \Cref{eq: convexity-condition-phi-0} is satisified strictly, by \Cref{lem: convexity-condition-Phi-0}, $\Phi_0$ is strictly convex, the minimiser is unique, and the convergence follows from the continuity of $\Phi_0$.
\end{proof}

\subsubsection*{\textbf{Non-convex Potential Function Cases}} If $\Phi_0$ is not convex, global convergence is not possible in general. We can show the following local convergence.

\begin{theorem}[Local Convergence]
    Assume that $r< 1$. Suppose the interior strict local minimiser $\bbs^\sharp$ of $\Phi_0$ admits a $(\gamma, \kappa)$-smooth neighbourhood (see \Cref{def: smooth-neighbourhood}) 
    $\Omega\left(\bbs^\sharp\right)$ such that $\frac{1}{1-r} \leq \frac{2\gamma}{\kappa^2}$, then the ER dynamic with initialisation $\bbs^0 \in \Omega\left(\bbs^\sharp\right)$ satisfies
    \begin{equation}
     0\leq \Phi_0\left(\bbs^T\right) - \Phi_0\left(\bbs^\sharp\right)  \leq  \frac{(1-r)\cdot\mathrm{KL}(\bbs^\sharp, \bbs^0)}{ T}, \label{eq: nonconvex-convergence-rate-ER}
    \end{equation}
    and $\lim_{T\rightarrow\infty} \bbs^T = \bbs^\sharp$. Suppose the interior strict local minimiser $\bbs^\flat$ admits a $(\gamma, \kappa)$-smooth neighbourhood $\Omega\left(\bbs^\flat\right)$ such that $1 \leq \frac{2\gamma}{\kappa^2}$, then the PR dynamic with initialisation $\bbs^0 \in \Omega\left(\bbs^\flat\right)$ satisfies
    \begin{equation}
     0\leq \Phi_0\left(\bbs^T\right) - \Phi_0\left(\bbs^\flat\right)  \leq  \frac{\mathrm{KL}(\bbs^\flat, \bbs^0)}{ T}, \label{eq: nonconvex-convergence-rate-PR}
    \end{equation}
    and $\lim_{T\rightarrow\infty} \bbs^T = \bbs^\flat$. \label{thm: local-convergence-dynamic}
\end{theorem}

The above theorem indicates that, in the cases that the potential function has various local minimizers, the limit point of the dynamic depends on its initialisation and the corresponding learning rate. Also, since $1 < \frac{1}{1-r}$, the local minimum attracting the ER dynamic is also attracting the PR dynamic. However, with a larger learning rate, the ER dynamic is able to escape the attractive basins. In the context of two-sided attention markets, if the content creators are reacting to the market change slower, the market dynamic may escape from a TSTOME near the initialisation, and converge to other TSTOMEs.

If $\Phi_0$ is concave, the minimising dynamic may converge to the boundary quickly. We present some special cases that the limit of the market dynamic is identifiable.
\begin{prop}
    For the both of the ER and PR dynamic, for some $i,j\in \calC$, suppose $rs_j\cdot\zeta_j'(s_j) + (r-1)\zeta_j(s_j) > 0$. And also $s_j^0 \geq s_i^0$, $\zeta_j(x) > \zeta_i(x)$, $v_j > v_i$, then it holds that 
    \[
     \lim_{t\rightarrow \infty} s_i^t = 0~. \label{prop: converge-to-boundary-dominance-prop}
    \] 
\end{prop}
\Cref{prop: converge-to-boundary-dominance-prop} states that when the potential function is concave, if one content creator is ``dominated'' by another creator in the sense that it has worse initialisation, larger cost of production and lower visibility, then the market share of this content creator will converge to $0$. Fixing the response functions $\zeta_j$, there exists sufficiently large $r$ such the potential function is concave. Hence, a stronger social influence signal implies the compression of weak content creators. 
A direct corollary of \Cref{prop: converge-to-boundary-dominance-prop} is: if a creator dominates all of the other creators, then they will become the monopoly of the market.
\begin{corollary}
    For both the ER and the PR dynamic, for some $j\in \calC$, suppose $rs_j\cdot\zeta_j'(s_j) + (r-1)\zeta_j(s_j) > 0$. And also $s_j^0 \geq s_i^0$, $\zeta_j(x) > \zeta_i(x)$, $v_j > v_i$ for every $i \in \calC \neq j$, then it holds that 
    \[
     s_j^t \rightarrow 1,~~ \text{as }t\rightarrow\infty. \label{cor: converge-to-boundary-dominance-prop}
    \] 
\end{corollary}

\subsection{Dynamics for a Set of Recommendation Strategies} 
\label{subsec: result-recommendation-policy}

The results in the previous two subsections can be generalized to settings where the RS employ popularity-ranking, quality-ranking, or a mixture of them. 
In these cases, we define $\bbs^t$ as follows:
\begin{equation}
s_j^t = \frac{v_j^{t+1}(\phi_j^t)^r}{\sum_i v_i^{t+1}(\phi_i^t)^r}. \label{eq: variable-transformation-ranking}
\end{equation}
We summarize the update rules of $\bbs^{t+1}$, the potential functions and the equivalences to mirror descent in a table below.
ER and PR are the popularity updates. ``$\equiv$ MD'' is an acronym for ``equivalent to mirror descent with KL divergence''.
$\eta$ refers to the learning rate. The update rule of $\bbs^{t+1}$ for ER and mixed ranking is
\begin{equation}
s_j^{t+1} ~=~ \frac{(\mu_j)^{\frac{1}{1-r}}\left(\zeta_j(s_j^{t})\right)^{\alpha+ \frac{r+\beta}{1-r}}\left(\zeta_j(s_j^{t-1})\right)^{\frac{\alpha r}{1-r}}(s_j^t)^{\frac{\beta}{1-r}}}{\sum_i (\mu_i)^{\frac{1}{1-r}}\left(\zeta_i(s_i^{t})\right)^{\alpha+ \frac{r+\beta}{1-r}}\left(\zeta_i(s_i^{t-1})\right)^{\frac{\alpha r}{1-r}}(s_i^t)^{\frac{\beta}{1-r}}}~,\label{eq:s-update-ER}
\end{equation}
while the update rule of $\bbs^{t+1}$ for PR and mixed ranking is
\begin{equation}
s_j^{t+1} ~=~ \frac{\mu_j \left(\zeta_j(s_j^t)\right)^{r+\alpha+\beta} (s_j^t)^{r+\beta}}{\sum_i \mu_i \left(\zeta_i(s_i^t)\right)^{r+\alpha+\beta} (s_i^t)^{r+\beta}}~.\label{eq:s-update-PR}
\end{equation}

\medskip

\noindent\hspace*{-0.15in}
\begin{tabular}{|l|l||c|c|c|}
\hline
\multicolumn{2}{|l||}{visibility} & \multirow{2}{*}{\textbf{popular ranking, $\calV_{\mathrm{pop}}$}} & \multirow{2}{*}{\textbf{quality ranking, $\calV_{\mathrm{qual}}$}} & \multirow{2}{*}{\textbf{mixed ranking, $\calV_{\mathrm{mix}}$}} \\
\multicolumn{2}{|l||}{updates} & & & \\
\hline
\multicolumn{2}{|l||}{} & \multirow{8}{*}{{\small $\begin{aligned} &\Phi_{\mathrm{pop}}(\bbs) =\\
&-\sum_{j} \Big[s_j \log \mu_j \\ &~~~+(r+\beta) \int_{0}^{s_j} \log \zeta_j(z) ~\mathrm{d}z\\ &~~~+(r+\beta-1)s_j\log s_j\Big]  \end{aligned}$}}
& \multirow{8}{*}{{\small $\begin{aligned} &\Phi_{\mathrm{qual}}(\bbs) =\\
&-\sum_{j} \Big[s_j \log \mu_j \\ &~~~+(r+\alpha) \int_{0}^{s_j} \log \zeta_j(z) ~\mathrm{d}z\\ &~~~+(r-1)s_j\log s_j\Big]  \end{aligned}$}}
& \multirow{8}{*}{{\small $\begin{aligned} &\Phi_{\mathrm{mix}}(\bbs) =\\
&-\sum_{j} \Big[s_j \log \mu_j \\ &~~~+(r+\alpha+\beta) \int_{0}^{s_j} \log \zeta_j(z) ~\mathrm{d}z\\ &~~~+(r+\beta-1)s_j\log s_j\Big]  \end{aligned}$}}\\
\multicolumn{2}{|l||}{} & & & \\
\multicolumn{2}{|l||}{} & & & \\
\multicolumn{2}{|l||}{potential} & & & \\
\multicolumn{2}{|l||}{function} & & & \\
\multicolumn{2}{|l||}{} & & & \\
\multicolumn{2}{|l||}{} & & & \\
\multicolumn{2}{|l||}{} & & & \\
\hline
\hline
\multirow{2}{*}{ER} & $\bbs$-update  & subs.~$\alpha=0$ into \eqref{eq:s-update-ER} & subs.~$\beta=0$ into \eqref{eq:s-update-ER} & see \Cref{eq:s-update-ER} \\
\cline{2-5}
& $\equiv$ MD? & YES, $\eta = \frac{1}{1-r}$ & \multicolumn{2}{|c|}{similar to MD with momentum; see remark (c) below} \\
\hline
\hline
\multirow{2}{*}{PR} & $\bbs$-update  & subs.~$\alpha=0$ into \eqref{eq:s-update-PR} & subs.~$\beta=0$ into \eqref{eq:s-update-PR} & see \Cref{eq:s-update-PR}  \\
\cline{2-5}
& $\equiv$ MD? & \multicolumn{3}{|c|}{YES, $\eta = 1$} \\
\hline
\end{tabular}

\medskip

We give five remarks about the three cases in the above table.

\noindent (a) When compared to $\Phi_0$, $\Phi_{\mathrm{pop}}$ is in exactly the same form after replacing $r$ with $r + \beta$. Additionally, the ER/PR dynamics retain the same learning rates as in the inactive platform $\calV_0$. By implementing the popularity ranking with difference choices of $\beta$, the platform can adjust the strength of the social signal without altering the core nature of the dynamics.

\noindent (b) Recall that for the constant policy case, the condition for $\Phi_0$ to be convex is $rs_j\cdot \zeta_j'(s_j) + (r-1)\zeta_j(s_j)\le 0$.
Then the condition for $\Phi_{\mathrm{pop}}$ to be convex is $(r+\beta)s_j\cdot \zeta_j'(s_j) + (r+\beta-1)\zeta_j(s_j) \le 0$, which is harder to satisfy,
since $rs_j\cdot \zeta_j'(s_j) + (r-1)\zeta_j(s_j) \leq (r+\beta)s_j\cdot \zeta_j'(s_j) + (r+\beta-1)\zeta_j(s_j)$.
Hence, whenever $\beta >0 $, the popularity ranking policy weakens the potential function's Bregman-smoothness or convexity, making it more difficult for the market dynamics to converge to its minimum. This aligns with the intuition that strong social signals drive polarization and unpredictability in market dynamics. Conversely, the platform may also use its power to weaken the social signal by setting $\beta <0$, thereby enhancing the stability of the dynamics.

\noindent (c) We cannot cast the $\bbs$-update for ER and quality/mixed rankings as mirror descent. However, this update still has some flavor analogous to \emph{mirror descent with momentum}.
An alternative form of the mirror descent update \Cref{eq: mirror-descent-update-rule} is
\begin{align*}
     \widetilde{\bbx^{t+1}} &= (\nabla h)^{-1}\left( \nabla h(\bbx^t) - \eta\cdot \nabla f(\bbx_t)\right), \\
     \bbx^{t+1} &= \argmin_{\bbx \in X} d_h\left(\bbx, \widetilde{\bbx^{t+1}}\right).
\end{align*}
Since $d_h$ is the KL-divergence, $(\nabla h)_j (\cdot) = \log(\cdot)$ and  $(\nabla h)_j^{-1} (\cdot) = \exp(\cdot)$ respectively. The update rule in \Cref{eq:s-update-ER} can be rewritten as
\begin{align*}
    \widetilde{\bbs^{t+1}} &= (\nabla h)^{-1}\left( \nabla h(\bbs^t) - \frac{\theta}{1-r}\cdot \nabla \Phi_{\mathrm{mix}}(\bbs^t)  - \frac{1-\theta}{1-r}\cdot \nabla \Phi_{\mathrm{mix}}(\bbs^{t-1})\right),\\
    \bbs^{t+1} &= \argmin_{\bbs \in \Delta} \mathrm{KL}\left(\bbs, \widetilde{\bbs^{t+1}}\right).
\end{align*}
where $\theta = \frac{r+\beta + \alpha - (r+\beta)\alpha}{r+\beta + \alpha }$. Compared with the standard mirror descent update, which is equivalent to gradient descent in the dual space, the above ER update can be seen as gradient descent with momentum in the dual space. Hence, ER and PR are, again, the minimising processes against the same potential function. 

\noindent (d) Performing the same transformation on $\Phi_{\mathrm{qual}}$ as \Cref{eq: interpretation-of-potential-0}, we get %
\begin{equation}
\Phi_{\mathrm{qual}}\left(\bbs^t\right) = \underbrace{\mathrm{KL}(\bbs^t, \boldsymbol{\mu})}_{\substack{\text{Attention alignment} \\ \text{to visibility}}} +  r\cdot\underbrace{H(\bbs^t)}_{\text{Entropy}}+ (r+\alpha)\cdot\left(\underbrace{\sum_j\int_{0}^{q_j^{t+1}} \frac{c'(u)}{u} ~\mathrm{d}u}_{\text{Cost of production}} - \underbrace{\sum_j s_j^t\log\left(q_j^{t+1}\right)}_{\text{Expected $\log$-quality}}\right). \label{eq: interpretation-of-potential-qual}
\end{equation}
With $\alpha >0$, by using quality ranking, compared to $\Phi_0$, the objective $\Phi_{\mathrm{qual}}$ puts more emphasis on the last two terms, which correspond to the market efficiency i.e. ``input cost of production - output qualities''. In contrast with popularity, which is equivalent to increasing $r$ in $\Phi_0$ globally, the quality ranking potential does not increase the weight on the entropy term. Hence, intuitively, the quality ranking policy is a more desirable recommendation policy such that it improves the market efficient without deteriorating the inequality of the attention distribution.

\noindent (e) By implementing the mixed ranking mechanism with appropriate parameters $\mu_j,\alpha,\beta$, we can recover any desired learning rates of mirror descent w.r.t.~the customised potential function $\Phi_{\mathrm{cust}}^{a,b, \boldsymbol{\sigma}}$. 
\begin{theorem}
    Define the customised potential as
    \[
    \Phi_{\mathrm{cust}}^{a,b, \boldsymbol{\sigma}}(\bbs) := -\left(\sum_{j} s_j \log \sigma_j + a \cdot\int_{0}^{s_j} \log \zeta_j(z) ~\mathrm{d}z + b\cdot s_j\log s_j  \right),
    \]
    where $a,b \in \rr$ are arbitrary constants. Then, with the social influence factor $r\in(0,1)$, by setting $\alpha = \eta a -\eta b + \eta -1$, $\beta = \eta b - \eta + 1 - r$ and $\mu_j = \sigma_j^\eta$, the PR update with mixed ranking strategy is equivalent to the mirror descent update with learning rate $\eta$. \label{thm: customised-potential-function}
\end{theorem}

\Cref{thm: customised-potential-function} indicates that 
the platform can control both the landscape of the potential function and the speed of market dynamics. Given some desired market equilibrium that the platform would like to reach, it is possible for the platform to reverse-engineer a mixed ranking strategy with appropriate parameters to control the behaviors and outcomes of market dynamics.%

\section{Empirical Results}
\begin{figure*}[bt]
  \centering

  \centering
  \includegraphics[width=1.\textwidth]{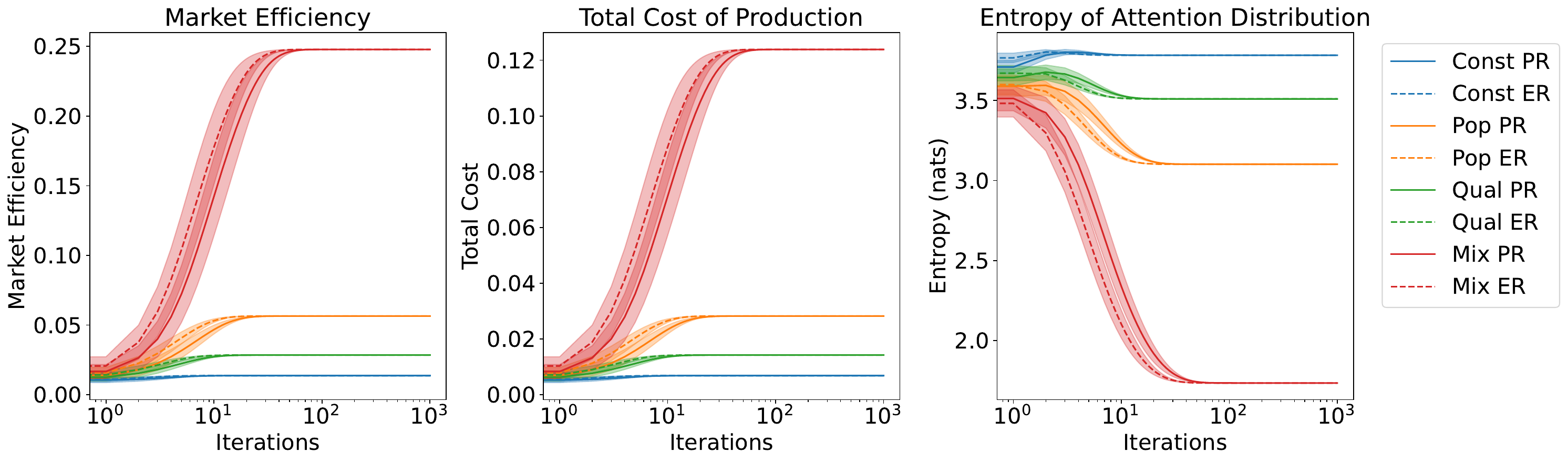}%

    \caption{Demonstration of simulation results with $r = 0.3, \alpha = \beta =0.1$. Each simulation is run for 1000 epochs. Each line is the average value of 100 runs. The shaded area indicates the standard deviations of the metrics across 100 runs. \emph{Left}: Market efficiency -- higher values indicates higher expected quality against the attention distribution $\bbs^t$, \emph{Middle}: Total production cost across all creators, \emph{Right}: Shannon Entropy of the attention distribution $\bbs^t$ -- higher values indicate a more uniform distribution. An entropy of zero means $s_j^t = 1$ for some $j$ and $s_i^t = 0$ for all $i \neq j$.}
    \label{fig:empiricalobs}
\end{figure*}

We simulate an instance of two-sided market with 50 producers under all four ranking strategies presented in \Cref{subsec: rs-policy} w.r.t. PR/ER dynamics. The potential function of each producer is set to be $c_j(q_j) = p_j \cdot q_j^2$ with $p_j$ sampled uniformly from $[0.5,5]$. We set $v_j = \mu_j = \frac{1}{50}$ for all $j\in \calC$. We select 100 initialisations $\bbs^0$ uniformly at random and report the average values of these 100 runs of the following three metrics:
\begin{itemize}
    \item \textbf{Market efficiency ($\sum_{j\in \calC} s_j^t q_j^t$)}: The expected quality under the attention distribution, equivalently the probability that a user at epoch  successfully purchases an item.
    \item \textbf{Total cost of production ($\sum_{j \in \calC} c_j(q_j^t)$)}: The aggregate input from all producers, indicating how strongly the market dynamics incentivize creator effort.
    \item \textbf{Shannon entropy ($-\sum_{j\in \calC} s_j^t\log s_j^t$)}: A measure of polarization in the attention distribution.
\end{itemize}

The results for $r=0.3$, $\alpha = \beta = 0.1$ are shown in \Cref{fig:empiricalobs}, while those for 
$r=0.2,0.4,0.5$ appear in \Cref{app-sec: additional-empirical}. In all cases, the three metrics converge after about $10^2$ iterations. Relative to their initial values, both the market efficiency and total production cost increase significantly, whereas the Shannon entropy decreases. This indicates that the dynamics enhance market efficiency and incentivize production but also polarize the attention distribution, consistent with the theoretical insights from the market’s potential functions (cf. \Cref{eq: interpretation-of-potential-0}).

From the left and middle panels of \Cref{fig:empiricalobs}, the mixed ranking strategy achieves the highest market efficiency, followed by the popularity and quality ranking strategies, while the constant ranking strategy performs the worst. This contrasts with generic trial-offer markets studied in \cite{van2016aligning,zhu2023stability}, where the quality ranking strategy outperforms all alternatives. The findings support our theoretical prediction that all ranking strategies enhance market efficiency yet polarize attention distribution. Since the mixed ranking strategy combines both quality and popularity factors, it achieves the greatest efficiency and lowest entropy.

As shown in the right panel of \Cref{fig:empiricalobs}, the Shannon entropy remains positive across all ranking strategies, indicating convergence to equilibria in which multiple producers hold positive market shares. Moreover, the standard deviations of these metrics converge to zero, demonstrating that different initializations lead to the same equilibrium under each ranking strategy—a reflection of dynamic stability. Finally, all ER dynamics converge faster than PR dynamics, affirming our theoretical result that ER corresponds to larger learning rates.

\section{Conclusion}
\label{sec:conclusion}
Our work provides an optimisation perspective and algorithmic validation of the famous invisible hand insight in economics --- that self-reinforcing markets can evolve towards desirable outcomes --- for two-sided attention markets. We introduce its first potential function, on which mirror descent corresponds to both user and creator updates. This potential function offers a meaningful interpretation as a combination of expected log-utility, cost of production, entropy, and the alignment between attention and visibility.
This novel approach paves the way for a deeper understanding of complex dynamics in online attention markets. Our findings suggest that the optimisation perspective could be extended to further explain and enhance these mechanisms, offering new insights into their efficiency and impact.

Future work can include incorporating personalised recommendations which can be represented as multi-dimensional quality values; incorporating a richer set of creator reward schemes and platform incentives; quantifying the fairness of reward allocation under two-sided dynamics.

\bibliographystyle{abbrvnat}
\bibliography{refs,marketdynamics}
\appendix

\section{Summary of Update Rules of Users, Recommender Systems and Content Creators}\label{app:update-summary}

In \Cref{sec:model}, we gave elaborate discussions of different update rules of users, recommender systems (RS) and content creators. We summarize the key notations and formulae here. Recall that $\bphi$ denotes popularity vector, $\bbv$ denotes visibility vector, $\bbq$ denotes quality vector, and $\bbc$ denotes the vector of cost functions of different content creators.

The general form of dynamical systems of two-sided attention markets is
\[
\begin{aligned}
\bphi^{t} &= \calP(\bbv^{t}, \bbq^{t}, \bphi^{t-1}) && \text{(users' popularity update)}\\
\bbv^{t+1} &= \calV(\bbq^{t}, \bphi^{t}, \bbv^{t}) && \text{(RS's visibility update)}\\
\bbq^{t+1} &= \calQ(\bbv^{t+1},\bphi^{t}, \bbc) && \text{(content creators' quality update)}
\end{aligned}
\]

\subsubsection*{\textbf{Popularity Update}} We examine the following two rules:
\begin{itemize}
\item Equilibrium Response (ER) Dynamic (\Cref{def:ER}), where
\[
\left[\calP_{\mathrm{ER}}(\bbv^{t}, \bbq^{t}, \bphi^{t-1})\right]_j := \frac{(q_j^t v_j^t)^{\frac{1}{1-r}}}{\sum_i (q_i^t v_i^t)^{\frac{1}{1-r}}}~.
\]
\item Proportional Response (PR) Dynamic (\Cref{def:PR}), where
\[\left[\calP_{\mathrm{PR}}(\bbv^{t}, \bbq^{t}, \bphi^{t-1})\right]_j := \frac{q_j^t v_j^t (\phi_j^{t-1})^r }{\sum_i q_i^t v_i^t (\phi_i^{t-1})^r }~.\]
\end{itemize}

\subsubsection*{\textbf{Visibility Update}} We examine the following two groups of rules for visibility update:
\begin{itemize}
\item Constant: $\calV_{0}(\bbq^{t}, \bphi^t, \bbv^t) := \bbv$, where each $v_j > 0$, and $\sum_j v_j = 1$.
\item Popularity-Quality mixed ranking: This includes a range of ranking mechanisms, specified by the parameters $\alpha$, $\beta$ and $(\mu_1,\mu_2,\ldots,\mu_{|C|})$ as below:
\[
\left[\calV_{\mathrm{mix}}(\bbq^{t}, \bphi^t, \bbv^t)\right]_j := \frac{\mu_j (q_j^t)^\alpha (\phi_j^t)^\beta}{\sum_i \mu_i (q_i^t)^\alpha (\phi_i^t)^\beta}~.
\]
In particular,
\begin{itemize}
\item When $\beta=1$ and $\alpha=0$, this is called popularity ranking.
\item When $\beta=0$ and $\alpha=1$, this is called quality ranking.
\end{itemize}
\end{itemize}

\subsubsection*{\textbf{Quality Update}} We examine best-response update. We need the assumptions that for every $j\in \calC$, $c_j$ is strictly convex and continuously differentiable, $c_j'(0) = 0$ and $c_j'(1)\ge 1$. Then the update rule can be explicitly written as
\[
\left[\calQ_{\mathrm{BR}}(\bbv^{t+1},\bphi^{t}, \bbc)\right]_j := \zeta_j\left(\frac{v_j^{t+1} (\phi_j^t)^r}{\sum_{i=1}^{|\creators|} v_i^{t+1} (\phi_i^t)^r}\right)~,
\]
where $\zeta_j := (c_j')^{-1}$.

\newpage
\section{Additional Empirical Results} \label{app-sec: additional-empirical}

\begin{figure*}[h!]
  \centering

  \centering
  \includegraphics[width=1.\textwidth]{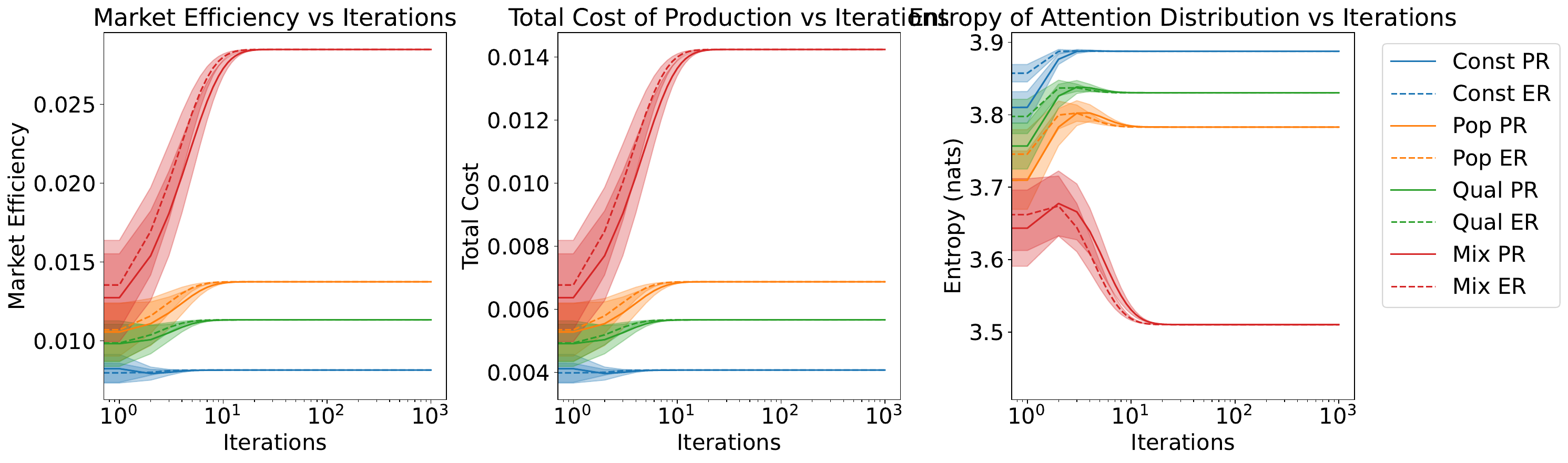}%

    \caption{Demonstration of simulation results with $r = 0.2, \alpha = \beta =0.1$. }

\end{figure*}

\begin{figure*}[h]
  \centering

  \centering
  \includegraphics[width=1.\textwidth]{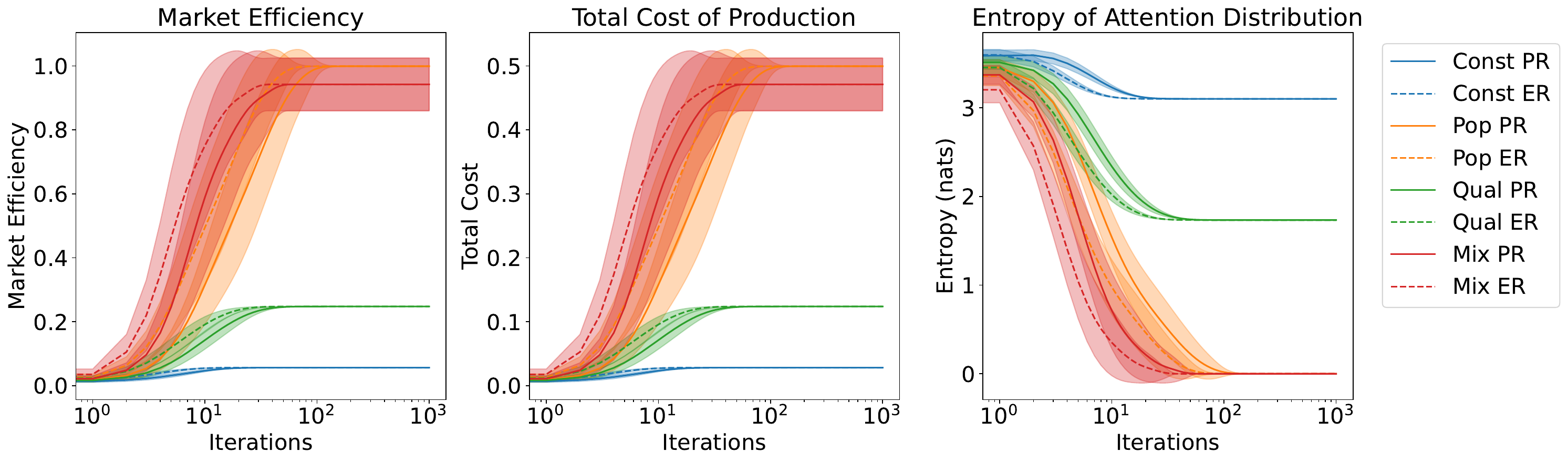}%

    \caption{Demonstration of simulation results with $r = 0.4, \alpha = \beta =0.1$. }
  
\end{figure*}

\begin{figure*}[h]
  \centering

  \centering
  \includegraphics[width=1.\textwidth]{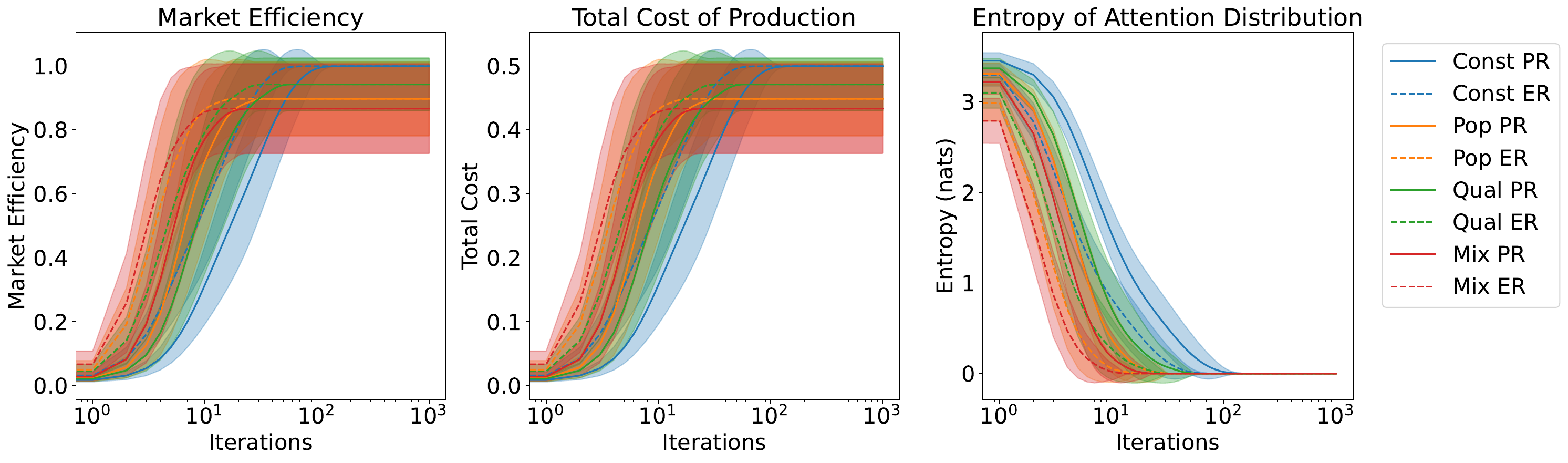}%

    \caption{Demonstration of simulation results with $r = 0.5, \alpha = \beta =0.1$. }
    
\end{figure*}

\section{Proofs for Section \ref{subsec: mirror-descent-results}}

\subsection{Proof of Theorem \ref{thm: mirror-descent-convergence-convex}}

\begin{lemma}[\cite{chen1993convergence}]
    If $\bbx^+$ is the optimal point for the optimization problem:
    \begin{align*}
        &\min_{\bbx}~~~~~~~~~~~~~~~~ g(\bbx) + d_h(\bbx,\bby), \\
        &\text{subject to     }\bbx\in C,
    \end{align*}
    where $g$ is a convex function and $C$ is a compact convex set. Then,
    \[
    g(\bbx) + d(\bbx,\bby) \geq g(\bbx^+) + d(\bbx^+, \bby) + d(\bbx, \bbx^+).
    \] \label{app-lem: difference-bregman-opt}
\end{lemma}
\begin{lemma}
    Suppose $\eta \leq \frac{1}{L}$, and the function $f$ is $L$-Bregman smooth. Then, consider the update \cref{eq: mirror-descent-update-rule}, we have
    \[
    f(\bbx^{t+1}) \leq f(\bbx^t).
    \] \label{app-lem: mirror-descent-always-decrease}
\end{lemma}
\begin{proof}
    By definition of Bregman smoothness, we have
    \[
     f(\bbx^{t+1}) \leq  f(\bbx^{t}) + \inner{\nabla f(\bbx^{t})}{\bbx^{t+1} - \bbx^{t} } + L \cdot d_h(\bbx^{t+1}, \bbx^{t}).
    \]
    Since $L \leq \frac{1}{\eta}$, by nonnegativity of the Bregman divergence, we have
    \begin{equation}
    f(\bbx^{t+1}) \leq  f(\bbx^{t}) + \inner{\nabla f(\bbx^{t})}{\bbx^{t+1} - \bbx^{t} } + \frac{1}{\eta} \cdot d_h(\bbx^{t+1}, \bbx^{t}). \label{app-eq: bregman-strong-convex-learning-rate}
    \end{equation}
    Then, by \cref{eq: mirror-descent-update-rule}, recall that
    \[
    \bbx^{t+1} = \argmin_{\bbx \in C} \left\{ \inner{\nabla f(\bbx^t)}{\bbx - \bbx^{t}} + \frac{1}{\eta} \cdot d_h(\bbx, \bbx^t)\right\},
    \]
    where the minimised objective is exactly the last two terms. Hence, 
    \[
    f(\bbx^{t+1}) \leq  f(\bbx^{t}) + \inner{\nabla f(\bbx^{t})}{\bbx^{t} - \bbx^{t} } + \frac{1}{\eta} \cdot d_h(\bbx^{t}, \bbx^{t}) = f(\bbx^{t}).
    \]
    
\end{proof}

\begin{proof}[Proof of \cref{thm: mirror-descent-convergence-convex}]
We note that, 
\begin{align*}
f(\bbx^{t+1})  &\leq f(\bbx^{t}) + \inner{\nabla f(\bbx^t)}{\bbx^{t+1} - \bbx^{t}} + \frac{1}{\eta}\cdot d_h(\bbx^{t+1},\bbx^{t}) \\
& \leq f(\bbx^{t}) + \inner{\nabla f(\bbx^t)}{\bbx^{\star} - \bbx^{t}} + \frac{1}{\eta}\cdot\left(d_h(\bbx^{\star},\bbx^{t}) - d_h(\bbx^{\star},\bbx^{t+1})\right) \\ 
& \leq f(\bbx^\star) + \frac{1}{\eta}\cdot\left(d_h(\bbx^{\star},\bbx^{t}) - d_h(\bbx^{\star},\bbx^{t+1})\right),
\end{align*}
where the first line follows from \cref{app-eq: bregman-strong-convex-learning-rate}, the second line follows from \cref{app-lem: difference-bregman-opt}, and the last line follows from the fact that $f$ is convex. Summing LHS and RHS from $t=0$ to $t = T-1$, we have
\[
T\cdot\left(f(\bbx^T) - f(\bbx^\star)\right) \leq  \sum_{t=0}^T f(\bbx^t) - f(\bbx^\star) \leq \frac{1}{\eta} \cdot d_h(\bbx^\star, \bbx^0) - d_h(\bbx^\star, \bbx^{T}),
\]
where the first inequality follows from \cref{app-lem: mirror-descent-always-decrease} and the second inequality is exactly the identity above after summed up. Note that $d_h$ is always positive, we could get the result by dividing both sides by $T$.
\end{proof}

\subsection{Proof of Theorem \ref{thm: mirror-descent-convergence-non-convex-detailed}}
Since the function $f$ is $L$-Bregman smooth. We must have
\begin{equation}
f(\bbx) - f(\bby) - \inner{\nabla f(\bby)}{\bbx - \bby} \leq L\cdot\mathrm{KL}(\bbx,\bby). \label{eq: RHS-bregman-convex}
\end{equation}
Hence, it is true that, for $\frac{1}{\eta} \geq L$, we have
\begin{align*}
f(\bbx^{t+1})  &\leq f(\bbx^{t}) + \inner{\nabla f(\bbx^t)}{\bbx^{t+1} - \bbx^{t}} + \frac{1}{\eta}\mathrm{KL}(\bbx^{t+1},\bbx^{t}) \\
&\leq f(\bbx^{t}) + \inner{\nabla f(\bbx^t)}{\bbx^{t} - \bbx^{t}} + \frac{1}{\eta}\mathrm{KL}(\bbx^{t},\bbx^{t}) =  f(\bbx^{t}),
\end{align*}
where the first inequality follows from \cref{eq: RHS-bregman-convex} and positivity of KL divergence, the second inequality follows from the optimality of $\bbx^{t+1}$ given by the mirror descent update rule. Now, consider the following lemma.
\boxlem{For some open ball $\mathcal{B}_{k}(x^\star)$ with $k > 0$, if $\bbx^\star$ is a local minimiser on $\mathcal{B}_{k}(x^\star) \cap \Delta$. Then, there exists $k' > 0$  such that $\bbx^\star$ is a local minimiser of the set $\Omega := \{\bbz \in \Delta:~ \mathrm{KL}(\bbx^\star, \bbz) < k'\}$.}
\begin{proof}
    For any point $\bbz \in \Omega$, set $k' = \frac{k^2}{2}$, we have
    \[
    \Vert \bbx^\star - \bbz \Vert_2 \leq  \Vert \bbx^\star - \bbz \Vert_1 \leq \sqrt{2\mathrm{KL}(\bbx^\star, \bbz)} \leq \sqrt{2 \cdot \frac{k^2}{2}}  = k,
    \]
    where the second inequality follows from Pinsker inequality. Hence, we have shown that $\Omega \subset \mathcal{B}_{k}(x^\star) \cap \Delta$. Since $k' > 0$, by continuity of KL divergence, we could also observe that $\Omega$ is open with respect to the subspace topology on $\Delta$.
\end{proof}

Next, we point out that, for any interior strict local minimiser $x^\star$, there exists a neighbourhood $\Omega$ of $x^\star$ such that $f(\bbx^\star) \geq  f(\bby) + \inner{\nabla f(\bby)}{\bby - \bbx^\star}$. Note that this result holds if one replace $f$ with any smooth function $f$.
\boxlem{
For some open set $\Omega \subset \rr^n$, suppose $\bbx^\star\in \mathrm{rint}\Delta$ is a strict regular local minimiser for some smooth function $f$ over $\Omega \cap \Delta$. Suppose $f$ has finitely many local minimisers in $\text{rint}\Delta$, there exists a open set $\Omega' \subset \Omega$ such that,
\begin{align}
f(\bbx^\star) -  f(\bby) - \inner{\nabla f(\bby)}{\bbx^\star -\bby} \geq  0~, \\
\inner{\nabla f(\bby)}{\bby - \bbx^\star} -\gamma \Vert \bby - \bbx^\star \Vert_2^2 \geq 0~, \\
\left\Vert \nabla f(\bby) - \nabla f(\bbx^\star)  \right\Vert_2 -\gamma\left\Vert \bby - \bbx \right\Vert_2\geq  0~.
\end{align}
for any $\bby \in \Omega' \cap \Delta$ and some $\gamma > 0$. \label{lem: local-convexity-around-strict-minimiser}
}
\begin{proof}
    Given any $\bbz \in \Omega$ close enough to $\bbx^\star$, consider the function $e(t) = f(\bbx^\star + t(\bbz - \bbx^\star))$.  By convexity of $\Delta$, the point $\bbz_t = \bbx^\star + t(\bbz - \bbx^\star) \in \Delta$ for any $t\in [0,1]$. And also, since $\bbz$ is close enough to $\bbx^\star$, $\bbz_t \in \Omega$ for $t\in [0,1+\alpha]$ when $\alpha>0$ is small enough. Also, since $\sum_{j} z_j - x^\star_j = 0$ and $\bbx^\star \in \text{rint}\Delta$, the point $\bbz_t \in \Omega$ for some $t \geq -\alpha$ when $\alpha >0$ and small enough. Therefore, since $\bbx^\star$ is the strict local minimiser, we shall also have
    \[
      \argmin_{t \in [-\alpha, 1+ \alpha]} e(t) = 0,~~~ e'(0) = 0, ~~~ e''(0) > 0.
    \]
    Hence, by continuity of $e''$, it is true that $e''(t) >\gamma >0$ whenever $|t|$ and $\gamma$ is small enough. Note that this $\gamma$ is invariant of the chosen direction by the regularity of the minimiser.
    Once $\bbz$ is chosen to be close enough to $\bbx^\star$, we can say that $e''(t) > \gamma$ for $t \in [0,1]$. Therefore, we must have
    \[
    f(\bbx^\star) = e(0) \geq e(t) + e'(t)(0-t) = f(\bbz_t) + \inner{\nabla f(\bbz_t)}{\bbx^\star - \bbz_t},
    \]
    as desired. For the second property, we know that
    \[
    \inner{\nabla f(\bbz_t)}{\bbz_t - \bbx^\star}=t\cdot(e'(t) -e'(0) )= t\cdot\int_{0}^t e''(t) ~\mathrm{d}t \geq t^2 \gamma = \frac{\Vert \bbz_t - \bbx^\star\Vert_2^2}{\Vert \bbz - \bbx^\star\Vert_2^2}\cdot \gamma.
    \]
    Since one can always choose $\bbz$ such that $\Vert \bbz - \bbx^\star\Vert_2^2 \leq c $ uniformly for some $c>0$, we have the second property. For the third property, we notice that, for small enough $|t|$
    \begin{align*}
    \left\Vert \nabla f(\bbz_t) - \nabla f(\bbx^\star)  \right\Vert \cdot \left\Vert  \bbz_t - \bbz^\star  \right\Vert &\geq \inner{\nabla f(\bbz_t) - \nabla f(\bbx^\star)}{\bbz_t - \bbx^\star }\\
    &= t\cdot(e'(t) - e'(0)) \geq  \gamma t^2 = \gamma \cdot \frac{\Vert \bbz_t - \bbx^\star \Vert^2}{\Vert \bby - \bbx^\star \Vert^2},
    \end{align*}
    where the first inequality is the Cauchy-Schwarz inequality, the second inequality follows from the fact that $e''(t) > \gamma$ whenever $t$ is small enough. Factoring out the term $\Vert \bbz_t - \bbz^\star \Vert_2$ and noting that $\Vert \bby - \bbx^\star\Vert_2$ could chosen to be bounded will yield the result. 
\end{proof}
Next, we show that, whenever the dynamic is near the strict local minimiser, it will not escape from this local region.
\boxlem{
Let $\bbx^\star$ be an interior local minimiser of a smooth function $f$ on $\Delta$. Under the PR dynamic, suppose $\bbx^{t}$ is close enough to $\bbx^\star$, and $\eta \leq \frac{2\gamma}{\kappa^2}$ then,
\[
\mathrm{KL}(\bbx^\star, \bbx^{t+1}) \leq \mathrm{KL}(\bbx^\star, \bbx^{t}) .
\]
\label{lem: no-escape-from-the-local-region}
}
\begin{proof}
    For notational simplicity, we write $g_j^t = \frac{\partial f}{\partial x_j}\Big|_{x_j = x_j^t}$ and  $g_j^\star = \frac{\partial f}{\partial x_j}\Big|_{x_j = x_j^\star}$. 
    Consider the gradient at the interior local minimiser, by KKT theorem, one must have
    \[
    g_j^\star  - \lambda = 0,
    \]
    for some $\lambda$. Next, we note that
    \begin{align*}
    \mathrm{KL}(\bbx^\star,  \bbx^{t+1}) - \mathrm{KL}(\bbx^\star,  \bbx^{t}) &=  \sum_{j}x_j^\star\log\left(\frac{\sum_i x_i^t \exp(-\eta g_i^t)}{x_j^t \exp(-\eta g_j^t)} \cdot \frac{x_j^t}{x_j^\star}\right)\\
    &= \sum_j \eta x_j^\star g_j^t + \log\left(\sum_i x_i^t \exp(-\eta g_i^t)\right) \\
    &= \sum_j \eta x_j^\star g_j^t -\eta\lambda + \log\left(\sum_i x_i^t \exp(-\eta g_i^t)\right) + \eta\lambda \\
    \end{align*}
    By continuity of $\nabla f$, in this local region, there exists some $\kappa>0$ such that 
    \[
    \left\Vert \nabla f(\bbx^t) - \nabla f(\bbx^t) \right\Vert_2 \leq \kappa\cdot \Vert \bbx^t - \bbx^\star \Vert_2.
    \]
    With this, let's look at the first term of the above difference of KL divergence, we have
    \begin{align*}
        \eta \cdot \sum_j  x_j^\star g_j^t &= \eta \cdot \sum_j  (x_j^\star - x_j^t + x_j^t)  g_j^t \\
        &= \eta \cdot \sum_j  x_j^t  g_j^t + \eta \cdot \sum_j g_j^t(x_j^\star - x_j^t)\\
        &= \eta \cdot \sum_j  x_j^t  g_j^t + \eta \cdot \inner{\nabla f(\bbx^t)}{\bbx^\star - \bbx^t} \\
        &\leq \eta \cdot \sum_j  x_j^t  g_j^t - \eta\gamma\cdot \Vert \bbx^\star - \bbx^t \Vert_2^2. \\
        &\leq \eta \cdot \sum_j  x_j^t  g_j^t - \eta\gamma \cdot\frac{1}{\kappa^2}\cdot \Vert \nabla f(\bbx^\star) - f(\bbx^t)\Vert_2^2.
    \end{align*}
    Then, it is clear that 
    \begin{align*}
    &\mathrm{KL}(\bbx^\star,  \bbx^{t+1}) - \mathrm{KL}(\bbx^\star,  \bbx^{t})\\ 
    &= \eta \cdot \sum_j  x_j^\star g_j^t - \log\left(\sum_i x_i^t \exp(- \eta g_i^t)\right) \\
    &\leq  \eta \cdot \sum_j  x_j^t  g_j^t - \gamma \cdot\frac{\eta}{\kappa^2}\cdot \Vert \nabla f(\bbx^\star) - \nabla f(\bbx^t)\Vert_2^2 + \log\left(\sum_i x_i^t \exp(- \eta g_i^t)\right) \\
    &= \eta \cdot \sum_j  x_j^t  g_j^t - \eta\lambda - \gamma \cdot\frac{\eta}{\kappa^2}\cdot \Vert \nabla f(\bbx^\star) - \nabla f(\bbx^t)\Vert_2^2 + \log\left(\sum_i x_i^t \exp(- \eta g_i^t)\right) + \eta\lambda \\
    &= \eta \cdot \sum_j  x_j^t  g_j^t - \sum_jx_j^t\eta\lambda - \gamma \cdot\frac{\eta}{\kappa^2}\cdot \Vert \nabla f(\bbx^\star) - \nabla f(\bbx^t)\Vert_2^2 + \log\left(\sum_i x_i^t \exp(- \eta g_i^t)\right) + \log\left(\exp \left(\eta\lambda \right)\right)\\
    &= \eta \cdot \sum_j  x_j^t(g_j^t - \eta\lambda) - \gamma \cdot\frac{\eta}{\kappa^2}\cdot \Vert \nabla f(\bbx^\star) - \nabla f(\bbx^t)\Vert_2^2 + \log\left(\sum_i x_i^t \exp(\eta\lambda- \eta g_i^t)\right) \\
    &= \eta \cdot \sum_j  x_j^t(g_j^t - g_j^\star) - \gamma \cdot\frac{\eta}{\kappa^2}\cdot \Vert \nabla f(\bbx^\star) - \nabla f(\bbx^t)\Vert_2^2 + \log\left(\sum_i x_i^t \exp(\eta g_i^\star- \eta g_i^t)\right),
    \end{align*}
    where the last line utilises the first order optimality condition as discussed above.
    Set $y_j = g_i^\star - g_i^t$, and consider the function
    \[
    \Gamma(\bby) = \eta\sum_j  x_j^t(-y_j) - \gamma \cdot\frac{\eta}{\kappa^2}\cdot \Vert \bby\Vert_2^2 + \log\left(\sum_i x_i^t \exp(\eta y_i)\right).
    \]
    It suffices to show that $\Gamma(\bby) < 0$ when $\bby$ is near $\mathbf{0}$. It is clear that $\Gamma(\mathbf{0}) = \mathbf{0}$, consider the partial derivatives
    \[
    \frac{\partial \Gamma}{\partial y_j} = \eta x_j^t - \frac{2\eta \gamma}{\kappa^2}y_j + \frac{\eta x_j^t\exp(\eta y_j)}{\sum_i x_i^t \exp(\eta y_i)}.
    \]
    Again $\nabla \Gamma(\mathbf{0}) = \mathbf{0}$. For the second derivative, we have
    \[ 
    \frac{\partial^2 \Gamma}{\partial y_j^2}\Bigg|_{y_j = 0} = \frac{-2\eta \gamma}{\kappa^2} + \eta^2 x_j^t -\eta^2 (x_j^t)^2, ~~\frac{\partial^2 \Gamma}{\partial y_iy_j}\Bigg|_{y_j=0,y_i = 0} = -\eta^2 x_i^t x_j^t.
    \]
    Therefore, the Hessian could be written as
    \[
    \nabla^2\Gamma(\bby)\Big|_{\bby = \mathbf{0}} = -\eta \bbx^t \left(\bbx^t\right)^T + \mathrm{diag}\left(\frac{-2\eta \gamma}{\kappa^2} + \eta^2 x_j^t\right).
    \]
    Note that the only nonzero of $ -\eta \bbx^t \left(\bbx^t\right)^T $ is $-\eta \Vert \bbx^t \Vert_2^2$ which is negative. Hence, to make $\nabla^2\Gamma(\bby)\Big|_{\bby = \mathbf{0}}$ negative definite, it suffices to let $\frac{-2\eta \gamma}{\kappa^2} + \eta^2 x_j^t < 0$ for every $x_j^t$. Since $x_j^t \leq 1$, we just need $\frac{-2\eta \gamma}{\kappa^2} + \eta^2  <0$, which is exactly the requirement that $\eta < \frac{2\gamma}{\kappa^2}$. With this, we know that $\bby = 0$ is a strict local minimiser of $\Gamma(\bby)$, which indicates that $\Gamma(\bby) < 0$ whenever $\bby$ is close enough to $0$. This completes the proof.
\end{proof}

With the above, we are ready to show that, if initialised close enough to a local minimiser, the dynamic will converge to the local minimiser at the same rate. Suppose $\bbx^{t} \in \Omega' \cap \Delta$ indicated by \cref{lem: local-convexity-around-strict-minimiser} and \cref{lem: no-escape-from-the-local-region}, we shall have
\begin{align*}
f(\bbx^{t+1})  &\leq f(\bbx^{t}) + \inner{\nabla f(\bbx^t)}{\bbx^{t+1} - \bbx^{t}} + \frac{1}{\eta}\cdot\mathrm{KL}(\bbx^{t+1},\bbx^{t}) \\
& \leq f(\bbx^{t}) + \inner{\nabla f(\bbx^t)}{\bbx^{\star} - \bbx^{t}} + \frac{1}{\eta}\cdot\left(\mathrm{KL}(\bbx^{\star},\bbx^{t}) - \mathrm{KL}(\bbx^{\star},\bbx^{t+1})\right) \\ 
& \leq f(\bbx^\star) + \frac{1}{\eta}\cdot\left(\mathrm{KL}(\bbx^{\star},\bbx^{t}) - \mathrm{KL}(\bbx^{\star},\bbx^{t+1})\right),
\end{align*}
where the first step follows from \cref{eq: RHS-bregman-convex}, the second step follows from \cref{app-lem: difference-bregman-opt}, the third step follows from the local convexity indicated by \cref{lem: local-convexity-around-strict-minimiser}. Therefore, we are able to see that
\[
\frac{1}{\eta}\cdot\left(\mathrm{KL}(\bbx^{\star},\bbx^{t}) - \mathrm{KL}(\bbx^{\star},\bbx^{t+1}) \right)\geq f(\bbx^{t+1}) - f(\bbx^\star).
\]
Since, $\mathrm{KL}(\bbx^{\star},\bbx^{t}) - \mathrm{KL}(\bbx^{\star},\bbx^{t+1}) \geq  0$, by taking the sum from $t=0$ to $t=T-1$, one has
\[
\mathrm{KL}(\bbx^{\star},\bbx^{0}) - \mathrm{KL}(\bbx^{\star},\bbx^{T}) 
 = \sum_{t=0}^{T-1}\mathrm{KL}(\bbx^{\star},\bbx^{t}) - \mathrm{KL}(\bbx^{\star},\bbx^{t+1})
 \geq \eta T\cdot(f(\bbx^{T}) -f(\bbx^{\star}) ).
\]
Then, we could conclude that,
\[
  f\left(\bbx^{T}\right) -f\left(\bbx^{\star}\right)\leq \frac{\mathrm{KL}(\bbx^{\star},\bbx^{0})}{\eta T}.
\]
For local convergence, since $\bbx^\star$ is the unique strict local minimiser in $\Omega'$, one must have $\bbx^T \rightarrow \bbx^\star$ as $T\rightarrow \infty$.

\section{Proofs for Section \ref{subsec: results-equiv-to-MD}}
\subsection{Proof of Lemma \ref{lem: dynamics-formulation} }
\boxlem{[Proportional response dynamic]
Set $s_j^t = \frac{v_j(\phi_j^t)^r}{\sum_{i\in \mathcal{I}}v_i(\phi_i^t)^r}$ as the trial probability of item $j$. Let $\zeta_j = (c'_j)^{-1}$, the above dynamics admit the following update rules:
\[
\phi_{j}^{t+1} = \frac{s_j^t \zeta_j(s_j^t)}{\sum_{i\in \mathcal{I}}s_i^t\zeta_i(s_i^t)},~~~s_j^{t+1} = \frac{v_j (s_j^t\zeta_j(s_j^t))^r}{\sum_{i\in \calI}v_i (s_i^t\zeta_i(s_i^t))^r}.
\]
}
\begin{proof}
    We notice that,
    \[
    \phi^{t+1}_j = \frac{q_j^{t+1} v_j (\phi_j^t)^r}{\sum_{i}q_i^{t+1} v_i (\phi_i^t)^r} = \frac{ \zeta_j(s_j^t)v_j (\phi_j^t)^r}{\sum_{i}\zeta_i(s_i^t) v_i (\phi_i^t)^r} = \frac{ \zeta_j(s_j^t)\frac{v_j (\phi_j^t)^r}{\sum_k v_k (\phi_k^t)^r}}{\sum_{i}\zeta_i(s_i^t)  \frac{v_i (\phi_i^t)^r}{\sum_k v_k (\phi_k^t)^r}} = \frac{\zeta_j(s_j^t)s_j^t}{\sum_i \zeta_i(s_i^t)s_i^t},
    \]
    which is the LHS identity. For the RHS identity, we have
    \[
     s_j^{t+1} = \frac{v_j (\phi^{t+1}_j)^r}{\sum_{i}v_i (\phi^{t+1}_i)^r} = \frac{v_j \left(\frac{\zeta_j(s_j^t)s_j^t}{\sum_k \zeta_k(s_k^t)s_k^t}\right)^r}{\sum_i v_i \left(\frac{\zeta_j(s_i^t)s_i^t}{\sum_i \zeta_k(s_k^t)s_k^t}\right)^r} = \frac{v_j (s_j^t\zeta_j(s_j^t))^r}{\sum_{i}v_i (s_i^t\zeta_i(s_i^t))^r}.
    \]
\end{proof}

\boxlem{[Equilibrium response dynamic] Assume $r<1$, for the equilibrium response dynamic, we have
\[
 s_j^{t+1} = \frac{v_j^{\frac{1}{1-r}}\left(\zeta_j(s_j^t)\right)^{\frac{r}{1-r}}}{\sum_{i}v_i^{\frac{1}{1-r}}\left(\zeta_i(s_i^t)\right)^{\frac{r}{1-r}}}.
\]}
\begin{proof}
    This could shown by direct calculation. We have
    \[
    s_j^{t+1} = \frac{v_j(\phi_j^{t+1})^r}{\sum_i v_i(\phi_i^{t+1})^r} = \frac{v_j\left(\frac{(q_j^{t}v_j)^{\frac{1}{1-r}}}{\sum_k (q_k^{t} v_k)^{\frac{1}{1-r}}}\right)^r}{\sum_i v_i \left(\frac{(q_i^{t}v_i)^{\frac{1}{1-r}}}{\sum_k (q_k^{t} v_k)^{\frac{1}{1-r}}}\right)^r} = \frac{v_j^{\frac{1}{1-r}} \left(\zeta_j(s_j^t)\right)^{\frac{r}{1-r}}}{\sum_i v_i^{\frac{1}{1-r}} \left(\zeta_i(s_i^t)\right)^{\frac{r}{1-r}}}.
    \]
\end{proof}

\subsection{Proof of Theorem \ref{thm: dynamic-equiv-to-MD}}
\begin{proof}
    The form of mirror descent with KL divergence is
    \[
    s^{t+1}_j = \frac{s_j^t \exp (-\eta g_j^t)}{\sum_i s_i^t \exp (-\eta g_i^t)}.
    \]
    Hence, plugging in the gradient above, set $\eta = 1$, we have
    \[
    s_j^{t+1}= \frac{s_j^t \exp\left(\log(v_j[\zeta(s_j^t)]^r(s_j^t)^{r-1})\right)}{\sum_i s_i^t \exp\left(\log(v_i[\zeta(s_i^t)]^r(s_i^t)^{r-1})\right)} =  \frac{v_j (s_j^t\zeta_j(s_j^t))^r}{\sum_{i\in \calI}v_i (s_i^t\zeta_i(s_i^t))^r},
    \]
    as desired. Similarly, for $\eta = \frac{1}{1-r}$, we have
    \[
    s_j^{t+1}= \frac{s_j^t \exp\left(\frac{1}{1-r}\log(v_j[\zeta(s_j^t)]^r(s_j^t)^{r-1})\right)}{\sum_i s_i^t \exp\left(\frac{1}{1-r}\log(v_i[\zeta(s_i^t)]^r(s_i^t)^{r-1})\right)}
     = \frac{v_j^{\frac{1}{1-r}} \left(\zeta_j(s_j^t) \right)^{\frac{r}{1-r}}}{\sum_i v_i^{\frac{1}{1-r}} \left(\zeta_i(s_i^t) \right)^{\frac{r}{1-r}}}.
    \]
\end{proof}

\subsection{Transformation of $\Phi_0$} \label{app-sec: transformation-of-potential-function}
\begin{proof}
    For the integral in the second term of \cref{eq: potential-function-0}, consider the substitution $t = c_j'(z)$, we can write it as
    \[
    \int_{0}^{s_j} \log \zeta_j(z) ~\mathrm{d}z = \int_{\zeta_j(0)}^{\zeta_j(s_j)} \log(t) c_j''(t)~\mathrm{d}t.
    \]
    Then, using integration by parts, we have
    \[
    \int_{\zeta_j(0)}^{\zeta_j(s_j)} \log(t) c_j''(t)~\mathrm{d}t= s_j\log(\zeta_j(s_j)) - \int_{\zeta_j(0)}^{\zeta_j(s_j)} \frac{c_j'(u)}{u} ~\mathrm{d}u = s_j\log(\zeta_j(s_j)) - \int_{0}^{\zeta_j(s_j)} \frac{c_j'(u)}{u} ~\mathrm{d}u
    \]
    Replacing $s_j$ by $s_j^t$ and $\zeta_j(s_j^t)$ by $q_j^{t+1}$, we have
    \[
    \int_{0}^{s_j} \log \zeta_j(z) ~\mathrm{d}z = s_j^t\log(q_j^{t+1}) - \int_{0}^{q_j^{t+1}} \frac{c_j'(u)}{u} ~\mathrm{d}u.
    \]
\end{proof}

\section{Proofs for Section \ref{subsec: result-convergence-two-sided-dynamic}}
\subsection{Proof of Lemma \ref{lem: Bregman-smooth-phi0}}
\begin{proof} 
    Note that,
    \[
    \Phi(\bbx) - \Phi(\bby) = \sum_j (y_j - x_j)\log v_j + r\int_{x_j}^{y_j} \log \zeta_j(z) ~\mathrm{d}z + (1-r)x_j\log x_j - (1-r)y_j \log y_j + (r-1)(x_j - y_j).
    \]
    And also, 
    \[
    \inner{\nabla \Phi(\bby)}{\bbx - \bby} = \sum_{j}  (y_j - x_j)\log v_j- r(x_j - y_j)\log \zeta_j(y_j) - (r-1)\log(y_j)(x_j - y_j).
    \]
    Hence,
   \begin{align*}
   & \Phi(\bbx) - \Phi(\bby) - \inner{\nabla \Phi(\bby)}{\bbx - \bby} \\
    &= \sum_{j} r\int_{x_j}^{y_j} \log \zeta_j(z) ~\mathrm{d}z + r(x_j - y_j)\log \zeta_j(y_j) +  (1-r)x_j\log x_j - (1-r)x_j\log y_j\\
    &\leq \sum_j (1-r)x_j\log x_j - (1-r)x_j\log y_j = (1-r)\mathrm{KL}(\bbx,\bby),
    \end{align*}
    where the inequality follows from the fact that $\zeta_j$ is increasing. 
\end{proof}

\subsection{Proof of Proposition \ref{prop: converge-to-boundary-dominance-prop}}
We prove some special cases that the dynamic would converge to the boundary when some $\zeta_j$ is concave.
\boxlem{For the PR dynamic, suppose $rx\cdot\zeta_j'(x) + (r-1)\zeta_j(x) > 0$, $s_j^0 \geq s_i^0$, $\zeta_j(x) > \zeta_i(x)$, $v_j > v_i$, then it holds that 
\[
 s_i^t \rightarrow 0,~~ \text{as }t\rightarrow\infty.
\]}
\begin{proof}
    Under the PR dynamic, we observe that 
    \[
    \frac{s_j^{t+1}}{s_i^{t+1}} = \frac{\frac{v_j (s_j^t\zeta_j(s_j^t))^r}{\sum_{k\in \calI}v_k (s_k^t\zeta_i(s_k^t))^r}}{\frac{v_i (s_i^t\zeta_i(s_i^t))^r}{\sum_{k\in \calI}v_k (s_k^t\zeta_i(s_k^t))^r}}=  \frac{v_j}{v_i} \cdot \left(\frac{s_j^t}{s_i^t}\right)^r \cdot \left(\frac{\zeta_j(s_j^t)}{\zeta_i(s_i^t)}\right)^{r}.
    \]
    As discussed in \cref{eq: concave-condition}, the concavity condition suggests, for any $x_1 \geq x_2$, we have $\left(\frac{\zeta_j(x_1)}{\zeta_j(x_2)}\right)^r \geq \left(\frac{x_1}{x_2}\right)^{1-r} $. Hence, utilising the fact that $\zeta_j(x) > \zeta_i(x)$, we have
    \[
    \frac{s_j^{t+1}}{s_i^{t+1}} = \frac{v_j}{v_i} \cdot \left(\frac{s_j^t}{s_i^t}\right)^r \cdot \left(\frac{\zeta_j(s_j^t)}{\zeta_i(s_i^t)}\right)^{r} \geq \frac{v_j}{v_i} \cdot \left(\frac{s_j^t}{s_i^t}\right)^r \cdot \left(\frac{\zeta_j(s_j^t)}{\zeta_j(s_i^t)}\right)^{r} \geq \frac{v_j}{v_i} \cdot \frac{s_j^t}{s_i^t}.
    \]
\end{proof}
    
\boxlem{For the equilibrium response dynamic, suppose $rx\cdot\zeta_j'(x) + (r-1)\zeta_j(x) > 0$, $s_j^0 \geq s_i^0$, $\zeta_j(x) > \zeta_i(x)$, $v_j > v_i$, then it holds that 
\[
 s_i^t \rightarrow 0,~~ \text{as }t\rightarrow\infty.
\]}
\begin{proof}
    Again, we have
    \[
     \frac{s_j^{t+1}}{s_i^{t+1}} = \frac{ \frac{v_j^{\frac{1}{1-r}}\left(\zeta_j(s_j^t)\right)^{\frac{r}{1-r}}}{\sum_{k}v_k^{\frac{1}{1-r}}\left(\zeta_i(s_k^t)\right)^{\frac{r}{1-r}}}}{ \frac{v_i^{\frac{1}{1-r}}\left(\zeta_i(s_i^t)\right)^{\frac{r}{1-r}}}{\sum_{k}v_k^{\frac{1}{1-r}}\left(\zeta_k(s_k^t)\right)^{\frac{r}{1-r}}}}
     =\left(\frac{v_j}{v_i}\right)^{\frac{1}{1-r}}\cdot \left(\frac{\zeta_j(s_j^t)}{\zeta_i(s_i^t)}\right)^{\frac{r}{1-r}}.
    \]
    Then, by the concavity condition and the fact that $\zeta_j(x) > \zeta_i(x)$, we have
    \[
    \frac{s_j^{t+1}}{s_i^{t+1}} = \left(\frac{v_j}{v_i}\right)^{\frac{1}{1-r}}\cdot \left(\frac{\zeta_j(s_j^t)}{\zeta_i(s_i^t)}\right)^{\frac{r}{1-r}} \geq \left(\frac{v_j}{v_i}\right)^{\frac{1}{1-r}}\cdot \left(\frac{\zeta_j(s_j^t)}{\zeta_j(s_i^t)}\right)^{\frac{r}{1-r}} \geq \left(\frac{v_j}{v_i}\right)^{\frac{1}{1-r}}\cdot \frac{s_j^t}{s_i^t}.
    \]
\end{proof}

\section{Proofs for Section \ref{subsec: result-recommendation-policy}}
\subsection{Popularity Ranking}
\begin{lemma}[dynamics with popularity ranking]
 The update rule for the popularity ranking strategy $\calV_{\mathrm{pop}}$ of ER dynamic can be written as \[
  s_j^{t+1} = \frac{\left(\zeta_j(s_j^t)\right)^{\frac{r+\beta}{1-r}} \left( s_j^t\right)^{\frac{\beta}{1-r}} \cdot(\mu_j)^{\frac{1}{1-r}}}{\sum_i \left(\zeta_i(s_i^t)\right)^{\frac{r+\beta}{1-r}} \left( s_i^t\right)^{\frac{\beta}{1-r}} \cdot(\mu_i)^{\frac{1}{1-r}}}, ~~~~~~\forall j \in \calC.
 \] 
 The update rule for mixed ranking strategy $\calV_{\mathrm{mix}}$ of PR dynamic can be written as
    \[
    s_j^{t+1} = \frac{\mu_j \left(s_j^t \zeta_j(s_j^t)\right)^{r+\beta}}{\sum_i \mu_i \left(s_i^t \zeta_i(s_i^t)\right)^{r+\beta}}, ~~~~~~\forall j \in \calC.
    \] \label{lem: dynamic-formulation-pop-rec}

\end{lemma}
\begin{proof}
   For the ER dynamic: First, we note that the trial probability is now given by
   \[
    s_j^{t} = \frac{v_j^{t+1} (\phi_j^t)^r}{\sum_i v_i^{t+1} (\phi_i^t)^r} = \frac{\mu_j (\phi_j^t)^{r+\beta}}{\sum_i \mu_i (\phi_i^t)^{r+\beta}}.
   \]
    Given the update rule for the ER dynamic, under mixed recommendation policy, we have
\begin{align*}
   \phi^{t+1}_j &=  \frac{(q_j^{t+1} v_j^{t+1})^{\frac{1}{1-r}}}{\sum_{i =1}^{|\calI|} (q_i^{t+1} v_i^{t+1})^{\frac{1}{1-r}}} =  \frac{\left(\zeta_j(s_j^t) \mu_j(\phi_j^t)^\beta \right)^{\frac{1}{1-r}}}{\sum_{i =1}^{|\calI|} (q_i^t v_i^t)^{\frac{1}{1-r}}}
  = \frac{\left( \left(\zeta_j(s_j^t)\right) \mu_j(\phi_j^t)^\beta \right)^{\frac{1}{1-r}}}{\sum_{i =1}^{|\calI|} (q_i^t v_i^t)^{\frac{1}{1-r}}} \\
  &= \frac{\left(\zeta_j(s_j^t)\right)^{\frac{1}{1-r}} \left( \mu_j(\phi_j^t)^\beta \right)^{\frac{1}{1-r}}}{\sum_{i =1}^{|\calI|} (q_i^t v_i^t)^{\frac{1}{1-r}}} 
  = \frac{\left(\zeta_j(s_j^t)\right)^{\frac{1}{1-r}} \left( \mu_j(\phi_j^t)^{r+\beta} \right)^{\frac{1}{1-r}\cdot\frac{\beta}{r+\beta}} \cdot(\mu_j)^{\frac{r}{r+ \beta} \cdot\frac{1}{1-r}}}{\sum_{i =1}^{|\calI|} (q_i^t v_i^t)^{\frac{1}{1-r}}}\\
  & = \frac{\left(\zeta_j(s_j^t)\right)^{\frac{1}{1-r}} \left( \mu_j(\phi_j^t)^{r+\beta} \right)^{\frac{1}{1-r}\cdot\frac{\beta}{r+\beta}} \cdot(\mu_j)^{\frac{r}{r+ \beta} \cdot\frac{1}{1-r}}}{\sum_i\left(\zeta_i(s_i^t)\right)^{\frac{1}{1-r}} \left( \mu_i(\phi_i^t)^{r+\beta} \right)^{\frac{1}{1-r}\cdot\frac{\beta}{r+\beta}} \cdot(\mu_i)^{\frac{r}{r+ \beta} \cdot\frac{1}{1-r}}}
  = \frac{\left(\zeta_j(s_j^t)\right)^{\frac{1}{1-r}} \left( s_j^t\right)^{\frac{1}{1-r}\cdot\frac{\beta}{r+\beta}} \cdot(\mu_j)^{\frac{r}{r+ \beta} \cdot\frac{1}{1-r}}}{\sum_i\left(\zeta_i(s_i^t)\right)^{\frac{1}{1-r}} \left( s_i^t\right)^{\frac{1}{1-r}\cdot\frac{\beta}{r+\beta}} \cdot(\mu_i)^{\frac{r}{r+ \beta} \cdot\frac{1}{1-r}}}.
\end{align*}
Next, we have
\begin{align*}
    s_j^{t+1} = \frac{\mu_j (\phi_j^{t+1})^{r+\beta}}{\sum_i \mu_i (\phi_i^{t+1})^{r+\beta}} 
    =  \frac{\mu_j\left(\zeta_j(s_j^t)\right)^{\frac{r+\beta}{1-r}} \left( s_j^t\right)^{\frac{\beta}{1-r}} \cdot(\mu_j)^{\frac{r}{1-r}}}{\sum_i \mu_i\left(\zeta_i(s_i^t)\right)^{\frac{r+\beta}{1-r}} \left( s_i^t\right)^{\frac{\beta}{1-r}} \cdot(\mu_i)^{\frac{r}{1-r}}}
    = \frac{\left(\zeta_j(s_j^t)\right)^{\frac{r+\beta}{1-r}} \left( s_j^t\right)^{\frac{\beta}{1-r}} \cdot(\mu_j)^{\frac{1}{1-r}}}{\sum_i \left(\zeta_i(s_i^t)\right)^{\frac{r+\beta}{1-r}} \left( s_i^t\right)^{\frac{\beta}{1-r}} \cdot(\mu_i)^{\frac{1r}{1-r}}}.
\end{align*}

For the PR dynamic: 
    Firstly, we note that
    \[
    \phi^{t+1}_j = \frac{q_j^{t+1} v_j^{t+1} (\phi_j^{t})^r }{\sum_{i=1}^{|\calI|} q_i^{t+1} v_i^{t+1} (\phi_i^{t})^r } = \frac{\zeta_j(s_j^t)s_j^t}{\sum_i \zeta_i(s_i^t)s_i^t}.
    \] 
    Then, we have
    \[
    s_j^{t+1} = \frac{v_j^{t+2}(\phi_j^{t+1})^r}{\sum_i v_i^{t+2}(\phi_i^{t+1})^r} = \frac{\mu_j (\phi_j^{t+1})^{r+\beta}}{\sum_i \mu_i (\phi_i^{t+1})^{r+\beta}} =\frac{\mu_j \left(s_j^t \zeta_j(s_j^t)\right)^{r+\beta}}{\sum_i \mu_i \left(s_i^t \zeta_i(s_i^t)\right)^{r+\beta}}.
    \]
\end{proof}

We demonstrate one instance of the equivalence to mirror descent here. For other dynamics mentioned in \Cref{subsec: result-recommendation-policy}, the technique is the same.
\begin{theorem}
Define the potential function $\Phi_{\mathrm{pop}}$ as
\begin{equation}
\Phi_{\mathrm{pop}}(\bbs) := -\left(\sum_{j} s_j \log \mu_j + (r+\beta) \cdot\int_{0}^{s_j} \log \zeta_j(z) ~\mathrm{d}z + (r+\beta-1)s_j\log s_j  \right). 
\end{equation}
    \begin{itemize}
    \item The ER dynamic stated in \Cref{lem: dynamic-formulation-pop-rec} is equivalent to mirror descent of $\Phi_{\mathrm{pop}}$ with KL divergence on $\Delta$. The corresponding learning rate $\eta =\frac{1}{1-r}$,
    \item The PR dynamics stated in \Cref{lem: dynamic-formulation-pop-rec} is equivalent to mirror descent of $\Phi_{\mathrm{pop}}$ with KL divergence on $\Delta$. The corresponding learning rate $\eta =1$.

\end{itemize}
\label{thm: dynamic-equiv-to-MD-pop}
\end{theorem} 
\begin{proof}
    The gradient of mirror descent is 
    \[
    g_j^t = \frac{\partial}{\partial s_j}\Phi_{\mathrm{pop}}(\bbs^t) = \log \left( v_j \left(\zeta_j(s_j^t)\right)^{r+\beta} (s_j^t)^{r+\beta -1}  \right) - (r+\beta -1)
    \]
    The form of mirror descent with KL divergence is
    \[
    s^{t+1}_j = \frac{s_j^t \exp (-\eta g_j^t)}{\sum_i s_i^t \exp (-\eta g_i^t)}.
    \]
    Hence, plugging in the gradient above, set $\eta = 1$, we have
    \[
    s_j^{t+1}= \frac{s_j^t \exp\left(\log(v_j[\zeta(s_j^t)]^{r +\beta}(s_j^t)^{r+\beta-1})\right)}{\sum_i s_i^t \exp\left(\log(v_i[\zeta(s_i^t)]^{r+\beta}(s_i^t)^{r+\beta-1})\right)} =  \frac{v_j (s_j^t\zeta_j(s_j^t))^{r+\beta}}{\sum_{i\in \calI}v_i (s_i^t\zeta_i(s_i^t))^{r+\beta}},
    \]
    as desired. Similarly, for $\eta = \frac{1}{1-r}$, we have
    \[
    s_j^{t+1}= \frac{s_j^t \exp\left(\frac{1}{1-r}\log(v_j[\zeta(s_j^t)]^{r +\beta}(s_j^t)^{r+\beta-1})\right)}{\sum_i s_i^t \exp\left(\frac{1}{1-r}\log(v_i[\zeta(s_i^t)]^{r+\beta}(s_i^t)^{r+\beta-1})\right)}
     = \frac{\left(\zeta_j(s_j^t)\right)^{\frac{r+\beta}{1-r}} \left( s_j^t\right)^{\frac{\beta}{1-r}} \cdot(\mu_j)^{\frac{1}{1-r}}}{\sum_i \left(\zeta_i(s_i^t)\right)^{\frac{r+\beta}{1-r}} \left( s_i^t\right)^{\frac{\beta}{1-r}} \cdot(\mu_i)^{\frac{1}{1-r}}}.
    \]
\end{proof}

\subsection{Quality Ranking}
Set the variable as $s_j^t = \frac{v_j^{t+1}(\phi_j^t)^r}{\sum_i v_i^{t+1}(\phi_i^t)^r}.$
\begin{lemma}
   The PR dynamic with popularity ranking update could be reformulated as
   \[
   s_j^{t+1} = \frac{\mu_j \cdot \left(\zeta_j(s_j^t)\right)^{r+\alpha} \cdot (s_j^t)^r}{\sum_i \mu_i \cdot \left(\zeta_i(s_i^t)\right)^{r+\alpha} \cdot (s_i^t)^r}.
   \] And the ER dynamic with popularity ranking could be written as
   \[
   s_j^{t+1} = \frac{\mu_j^{\frac{1}{1-r}}\cdot \left(\zeta_j(s_j^t)\right)^{\alpha + \frac{r}{1-r}}\cdot\left(\zeta_j(s_j^{t-1})\right)^{r\alpha}}{\sum_i \mu_i^{\frac{1}{1-r}}\cdot \left(\zeta_i(s_i^t)\right)^{\alpha + \frac{r}{1-r}}\cdot\left(\zeta_i(s_i^{t-1})\right)^{r\alpha}}.
   \]
\end{lemma}
\begin{proof}
    Firstly, we note that
    \[
    \phi^{t+1}_j = \frac{q_j^{t+1} v_j^{t+1} (\phi_j^{t})^r }{\sum q_i^{t+1} v_i^{t+1} (\phi_i^{t})^r } = \frac{\zeta_j(s_j^t)s_j^t}{\sum_i \zeta_i(s_i^t)s_i^t}.
    \]
    Then, we have
    \[
    s_j^{t+1} = \frac{v_j^{t+2}(\phi_j^{t+1})^r}{\sum_i v_i^{t+2}(\phi_i^{t+1})^r} = \frac{\mu_j \left(\zeta_j(s_j^t)\right)^\alpha \left(\zeta_j(s_j^t)\cdot s_j^t\right)^r}{\sum_i \mu_i \left(\zeta_i(s_i^t)\right)^\alpha \left(\zeta_i(s_i^t)\cdot s_i^t\right)^r}= \frac{\mu_j \cdot \left(\zeta_j(s_j^t)\right)^{r+\alpha} \cdot (s_j^t)^r}{\sum_i \mu_i \cdot \left(\zeta_i(s_i^t)\right)^{r+\alpha} \cdot (s_i^t)^r}.
    \]
    For the ER dynamic, we have
    \[
    s_j^{t+1} = \frac{v_j^{t+2}(\phi_j^{t+1})^r}{\sum_i v_i^{t+2}(\phi_i^{t+1})^r} = \frac{v_j^{t+2} \cdot (v_j^{t+1})^{\frac{r}{1-r}} \cdot (q_j^{t+1})^{\frac{r}{1-r}}}{\sum_i v_i^{t+2} \cdot (v_i^{t+1})^{\frac{r}{1-r}} \cdot (q_i^{t+1})^{\frac{r}{1-r}}}
    = \frac{\mu_j (q_j^{t+1})^\alpha \cdot \mu_j^{\frac{r}{1-r}}(q_j^t)^{\frac{r\alpha}{1-r}} \cdot (q_j^{t+1})^{\frac{r}{1-r}}}{\sum_i \mu_i (q_i^{t+1})^\alpha \cdot \mu_i^{\frac{r}{1-r}}(q_i^t)^{\frac{r\alpha}{1-r}} \cdot (q_i^{t+1})^{\frac{r}{1-r}}}.
    \]
    Reformulating the RHS and noticing that $q_j^{t+1} = \zeta_j(s_j^{t})$ will yield the result.
\end{proof}
\subsection{Mixed ranking}
\begin{lemma}
     The ER dynamic with mixed ranking update could be reformulated as
     \[
     \frac{\mu_j^{\frac{1}{1-r}}\left(\zeta_j(s_j^{t})\right)^{\alpha+ \frac{r+\beta}{1-r}}\left(\zeta_j(s_j^{t-1})\right)^{\frac{\alpha r}{1-r}}(s_j^t)^{\frac{\beta}{1-r}}}{\sum_i \mu_i^{\frac{1}{1-r}}\left(\zeta_i(s_i^{t})\right)^{\alpha+ \frac{r+\beta}{1-r}}\left(\zeta_i(s_i^{t-1})\right)^{\frac{\alpha r}{1-r}}(s_i^t)^{\frac{\beta}{1-r}}},
     \]
     And the PR dynamic with mixed ranking update could be reformulated as
     \[
     s_j^{t+1} = \frac{\mu_j \left(\zeta_j(s_j^t)\right)^{r+\alpha+\beta} (s_j^t)^{r+\beta}}{\sum_i \mu_i \left(\zeta_i(s_i^t)\right)^{r+\alpha+\beta} (s_i^t)^{r+\beta}}.
     \]
\end{lemma}
\begin{proof}
    For the ER dynamic, we have
    \begin{align*}
    s_j^{t+1} &= \frac{v_j^{t+2}(\phi_j^{t+1})^{r}}{\sum_i v_i^{t+2}(\phi_i^{t+1})^{r}}= \frac{\mu_j (q_j^{t+1})^{\alpha}(\phi_j^{t+1})^{r+\beta}}{\sum_i \mu_i (q_i^{t+1})^{\alpha}(\phi_i^{t+1})^{r+\beta}} =  \frac{\mu_j (q_j^{t+1})^{\alpha+ \frac{r+\beta}{1-r}}(v_j^{t+1})^{\frac{r+\beta}{1-r}}}{\sum_i \mu_i (q_i^{t+1})^{\alpha+ \frac{r+\beta}{1-r}}(v_i^{t+1})^{\frac{r+\beta}{1-r}}} \\
    &= \frac{\mu_j (q_j^{t+1})^{\alpha+ \frac{r+\beta}{1-r}}(v_j^{t+1})^{\frac{\beta}{1-r}}(v_j^{t+1})^{\frac{r}{1-r}}}{\sum_i \mu_i (q_i^{t+1})^{\alpha+ \frac{r+\beta}{1-r}}(v_i^{t+1})^{\frac{\beta}{1-r}}(v_i^{t+1})^{\frac{r}{1-r}}}
    = \frac{\mu_j (q_j^{t+1})^{\alpha+ \frac{r+\beta}{1-r}}(v_j^{t+1})^{\frac{\beta}{1-r}}\mu_j^{\frac{r}{1-r}}(\phi_j^t)^{\frac{\beta r}{1-r}} (q_j^t)^{\frac{\alpha r}{1-r}}}{\sum_i \mu_i (q_i^{t+1})^{\alpha+ \frac{r+\beta}{1-r}}(v_i^{t+1})^{\frac{\beta}{1-r}}\mu_i^{\frac{r}{1-r}}(\phi_i^t)^{\frac{\beta r}{1-r}} (q_i^t)^{\frac{\alpha r}{1-r}}}\\
    &= \frac{\mu_j^{\frac{1}{1-r}}(q_j^{t+1})^{\alpha+ \frac{r+\beta}{1-r}}(q_j^t)^{\frac{\alpha r}{1-r}}(v_j^{t+1})^{\frac{\beta}{1-r}}(\phi_j^t)^{\frac{\beta r}{1-r}}}{\sum_i \mu_i^{\frac{1}{1-r}}(q_i^{t+1})^{\alpha+ \frac{r+\beta}{1-r}}(q_i^t)^{\frac{\alpha r}{1-r}}(v_i^{t+1})^{\frac{\beta}{1-r}}(\phi_i^t)^{\frac{\beta r}{1-r}}}= \frac{\mu_j^{\frac{1}{1-r}}(q_j^{t+1})^{\alpha+ \frac{r+\beta}{1-r}}(q_j^t)^{\frac{\alpha r}{1-r}}(s_j^t)^{\frac{\beta}{1-r}}}{\sum_i \mu_i^{\frac{1}{1-r}}(q_i^{t+1})^{\alpha+ \frac{r+\beta}{1-r}}(q_i^t)^{\frac{\alpha r}{1-r}}(s_i^t)^{\frac{\beta}{1-r}}}.
    \end{align*}
    For the PR dynamic, again, we apply the identity
     \[
    \phi^{t+1}_j = \frac{q_j^{t+1} v_j^{t+1} (\phi_j^{t})^r }{\sum_i q_i^{t+1} v_i^{t+1} (\phi_i^{t})^r } = \frac{\zeta_j(s_j^t)s_j^t}{\sum_i \zeta_i(s_i^t)s_i^t}.
    \]
    Then, 
    \[
    s_j^{t+1} = \frac{v_j^{t+2}(\phi_j^{t+1})^{r}}{\sum_i v_i^{t+2}(\phi_i^{t+1})^{r}} = \frac{\mu_j (q_j^{t+1})^{\alpha}(\phi_j^{t+1})^{r+\beta}}{\sum_i \mu_i (q_i^{t+1})^{\alpha}(\phi_i^{t+1})^{r+\beta}} = \frac{\mu_j \left(\zeta_j(s_j^t)\right)^{r+\alpha+\beta} (s_j^t)^{r+\beta}}{\sum_i \mu_i \left(\zeta_i(s_i^t)\right)^{r+\alpha+\beta} (s_i^t)^{r+\beta}}.
    \]
\end{proof}

\end{document}